\documentclass[11pt,authoryear]{elsarticle}
\usepackage{mydef}
\usepackage[show]{ed}
\usepackage[boxed]{algorithm2e}
\usepackage{rotating,booktabs}
\usepackage{multirow}
\usepackage{tikz}
\usepackage{lipsum}
\usepackage{mathtools}

\allowdisplaybreaks

\usetikzlibrary{calc,decorations.markings}
\usetikzlibrary{arrows,patterns,positioning}
\usetikzlibrary{shapes.geometric, backgrounds}

\def\calT{{\cal T}}
\def\calQ{{\cal Q}}
\def\disD{{\cal D}}
\def\disQ{{\cal Q}}
\def\calA{{\cal A}}

\def\dgxt{{D_{\Gamma_{X(t)}}}}
\def\gxt{{\Gamma_{X(t)}}}

\def\bs{\hfill $\blacksquare$}

\allowdisplaybreaks[3]

\newcommand{\dd}{\partial}

\newcommand{\CO}[2]{C^{(1)}_{#2,#1}}
\newcommand{\CT}[2]{C^{(2)}_{#2,#1}}
\newcommand\pFq[3]{{}_{#1}F_{#2}\left(#3\right){}}

\DeclarePairedDelimiterX\MeijerM[3]{\lparen}{\rparen}%
{\begin{smallmatrix}#1 \\ #2\end{smallmatrix}\delimsize\vert\,#3}

\newcommand\MeijerG[8][]{%
G^{\,#2,#3}_{#4,#5}\MeijerM[#1]{#6}{#7}{#8}}

\SetAlCapSkip{1em}
\graphicspath{{./Graphs/}}

\begin{document}

\begin{frontmatter}

\title{Geometric Local Variance Gamma model}

\author{P.~Carr}
\ead{petercarr@nyu.edu}

\author{A.~Itkin\corref{cor1}}
\ead{aitkin@nyu.edu}

\address{Tandon School of Engineering, New York University, \\
12 Metro Tech Center, RH 517E, Brooklyn NY 11201, USA}

\begin{abstract}
This paper describes another extension of the Local Variance Gamma model originally proposed by P. Carr in 2008, and then further elaborated on by Carr and Nadtochiy, 2017 (CN2017), and Carr and Itkin, 2018 (CI2018). As compared with the latest version of the model developed in CI2018 and called the ELVG (the Expanded Local Variance Gamma model), here we provide two innovations. First, in all previous papers the model was constructed based on a Gamma time-changed {\it arithmetic} Brownian motion: with no drift in CI2017, and with drift in CI2018, and the local variance to be a function of the spot level only. In contrast, here we develop a {\it geometric} version of this model with drift. Second, in CN2017 the model was calibrated to option smiles assuming the local variance is a {\it piecewise constant} function of strike, while in CI2018 the local variance is a {\it piecewise linear} function of strike. In this paper we consider 3 {\it piecewise linear} models: the local variance as a function of strike, the local variance as function of log-strike, and the local volatility as a function of strike (so, the local variance is a {\it piecewise quadratic} function of strike). We show that for all these new constructions it is still possible to derive an ordinary differential equation for the option price, which plays a role of Dupire's equation for the standard local volatility model, and, moreover, it can be solved in closed form. Finally, similar to CI2018, we show that given multiple smiles the whole local variance/volatility surface can be recovered which does not require solving any optimization problem. Instead, it can be done term-by-term by solving a system of non-linear algebraic equations for each maturity which is fast.
\end{abstract}


\end{frontmatter}

\section{Introduction}

The Local Variance Gamma (LVG) volatility model was first introduced by P. Carr in 2008 and then presented in \cite{CarrLVGOrig2014,CarrNadtochiy2017} as an extension of the local volatility model by \cite{Dupire:94} and \cite{derman/kani:94}. The latter was developed on the top of the celebrating Black-Scholes model to take into account the existence of option smile. The main advantage of all local volatility models is that given European options prices or their implied volatilities at points $(T,K)$ where $K,T$ are the option strike and time to maturity, they are able to exactly replicate the local volatility function $\sigma(T,K)$ at these points. This process is called calibration of the local volatility (or, alternatively, implied volatility) surface, see survey in  \cite{ELVG,ItkinLipton2017} and references therein.

However, as compared with the classical local volatility model, the LVG and ELVG have several advantages. First, they are richer in the financial sense. Indeed, it is worth noting that the term "local" in the name of the LVG/ELVG models is a bit confusing. This is because, e.g., the ELVG is constructed by equipping an arithmetic Brownian motion with drift and local volatility by stochastic time change $\gxt$. Here $\Gamma_t$ is a Gamma stochastic variable, and $X(t)$ is a deterministic function of time $t$. As stochastic change is one of the ways of introducing stochastic volatility, it could be observed that the LVG/ELVG is actually a local stochastic volatility (LSV) model which combines local and stochastic features of the volatility process. For more information on the LSV models, see \cite{Bergomi2016,kienitz2012financial}.

Another advantage of the LVG/ELVG is that their calibration is computationally more efficient. This is because this construction gives rise not to a partial differential equation (which in the classical case is known as Dupire's equation), but to a partial differential difference equation (PDDE). The latter is actually an ordinary differential equation (ODE) and permits both explicit calibration and fast numerical valuation. In particular, calibration of the local variance surface does not require any optimization method, rather just a root solver, \cite{ELVG}.

As discussed in \cite{ItkinLipton2017}, given the market quotes of European options for various maturities and strikes, the local (and then implied) volatility surface can be obtained by directly solving the Dupire equation using either analytical or numerical methods. The advantage of such an approach is that it guarantees no-arbitrage if the corresponding analytical or numerical method does preserve no-arbitrage (including  various interpolations, etc.). Obviously, solving Dupire's PDE requires either numerical methods, e.g. that in \cite{Coleman2001}, or, as in \cite{ItkinLipton2017}, a semi-analytic method which: i) first uses the Laplace-Carson transform, and ii) then applies various transformations to obtain a closed form solution of the transformed equation in terms of Kummer Hypergeometric functions. Still, it requires an inverse Laplace transform to obtain the final solution.
To make the second approach tractable, some assumptions should be made about the behavior of the local/implied volatility surface at strikes and maturities where the market quotes are not known. Usually, the corresponding local variance is assumed to be either piecewise constant, \cite{LiptonSepp2011}, or piecewise linear \cite{ItkinLipton2017} in the log-strike space, and piecewise constant in the time to maturity space.
A similar assumption is also necessary to make the LVG/ELVG models tractable. In particular, in \cite{CarrNadtochiy2017} the model was calibrated to option smiles assuming the local variance is a {\it piecewise constant} function of strike, while in \cite{ELVG} the local variance is a {\it piecewise linear} function of strike.

Despite these nice features of the ELVG, one possible problem could be that the model is developed based on the arithmetic Brownian motion with drift. That means that the underlying, in principle, could acquire negative values, which in some cases is undesirable, e.g., if the underlying is a stock price. Therefore, in this paper we describe another extension of the LVG model which operates with a Gamma time-changed {\it geometric} Brownian motion with drift, and the local variance which is a function of the spot level only (so is not a function of time).

Second, in \cite{CarrNadtochiy2017} the model was calibrated to option smiles assuming the local variance is a {\it piecewise constant} function of strike, while in \cite{ELVG} the local variance is a {\it piecewise linear} function of strike. In this paper we consider
 3 {\it piecewise linear} models: the local variance as a function of strike, the local variance as a function of log-strike, and the local volatility as a function of strike (so, the local variance is a {\it piecewise quadratic} function of strike). We show that in this new model it is still possible to derive an ordinary differential equation for the option price, which plays a role of Dupire's equation for the standard local volatility model. Moreover, it all three cases, this equation can be solved in closed form. Finally, similar to \cite{ELVG} we show that given multiple smiles the whole local variance/volatility surface can be recovered which does not require solving any optimization problem. Instead, it can be done term-by-term, and for every maturity the entire calibration is done by solving a system of non-linear algebraic equations which is significantly faster.

The rest of the paper is organized as follows. In Section~\ref{Process} the new model, which for an obvious reason we call the Geometric Local Variance Gamma model or the GLVG, is formulated. In Section~\ref{ForwardEq} we derive a forward equation (which is an ordinary differential equation (ODE)) for Put option prices using a homogeneous Bochner subordination approach. Section~\ref{cases} generalizes this approach by considering the local variance being piecewise constant in time. A closed form solution of the derived ODE is given in terms of Hypergeometric functions for various models of the local variance or volatility. The next Section discusses computation of a source term of this ODE which requires a no-arbitrage interpolation. Using the idea of \cite{ItkinLipton2017}), we show how to construct non-linear interpolation which provides both no-arbitrage, and a nice tractable representation of the source term, so that all integrals in the source term can be computed in closed form. In Section~\ref{calib} calibration of multiple smiles in our model is discussed in detail. To calibrate a single smile we derive a system of nonlinear algebraic equations for the model parameters, and explain how to obtain a smart guess for their initial values.  In Section~\ref{numExp} we discuss the results of some numerical experiments where calibration of the model to the given market smiles is done term-by-term. The last Section concludes.

\section{Stochastic model} \label{Process}

Let $W_t$ be a $\mathbb{Q}$ standard Brownian motion with time index $t \ge 0$. Consider a stochastic process $D_t$ to be a time-homogeneous diffusion with drift $\mu$
\begin{equation} \label{D}
 d D_t = \mu D_t d t + \sigma(D_t) D_t d W_t,
\end{equation}
\noindent where the volatility function $\sigma$ is local and time-homogeneous.

A unique solution to \eqref{D} exists if $\sigma(D) : \mathbb{R} \to \mathbb{R}$ is
Lipschitz continuous in $D$ and satisfies growth conditions at infinity. Since $D$ is a time-homogeneous Markov process, its infinitesimal generator $\calA$ is given by
\begin{equation} \label{gen}
\calA \phi(D) \equiv \left[\mu D \nabla_D  + \frac{1}{2} \sigma^2(D) D^2 \nabla^2_D \right] \phi(D)
\end{equation}
\noindent for all twice differentiable functions $\phi$. Here $\nabla_x$ is a first order
differential operator on $x$. The semigroup of the $D$ process
(which here is an expectation under $\mathbb{Q}$) is
\begin{equation} \label{semig}
\calT^D_t \phi(D_t) = e^{t \calA} \phi(D_t)  = \EQ[\phi(D_t)|D_0 = D], \quad \forall t \ge 0.
\end{equation}
This first equality could be also thought of as the Feynman-Kac theorem representation of the solution to the terminal value problem (see, e.g., \cite{FCT2011}), which connects the expectation in the right hand side to the solution of the corresponding PDE, and then the formal solution of this PDE is given by the exponential operator $e^{t \calA}$ applied to the initial condition $\phi(D_t)$.

In the spirit of \cite{CarrNadtochiy2017,ELVG}, introduce a new process $D_{\Gamma_t}$ which is $D_t$ subordinated by the unbiased Gamma clock $\Gamma_t$. The density of the unbiased Gamma clock $\Gamma_t$ at time $t \ge 0$ is
\begin{equation} \label{gammaDen}
\mathbb{Q}\{\Gamma_t \in d\nu\} = \dfrac{\nu^{m-1} e^{-\nu m /t}}{(t^*)^m \Gamma(m)} d\nu, \quad \nu > 0, \quad m \equiv t/t^*.
\end{equation}
Here $t^* > 0$ is a free parameter of the process, $\Gamma(x)$ is the Gamma function. It is easy to check that
\begin{equation} \label{gammaExp}
\EQ[\Gamma_t] = t.
\end{equation}
\noindent Thus, on average the stochastic gamma clock $\Gamma_t$ runs synchronously with the calendar time $t$.

As applied to the option pricing problem, we introduce a more complex construction.
Namely, consider options written on the underlying process $S_t$. Without loss of generality and for the sake of clearness let us treat below $S_t$ as the stock price process. Let us define $S_t$ as
\begin{equation} \label{sub}
S_t = D_{\Gamma_{X(t)}}
\end{equation}
\noindent where $X(t)$ is a deterministic function of time $t$. We need to determine $X(t)$ such that under a risk-neutral measure $\mathbb{Q}$, the total gains process $\hat{S}_t$, including the underlying price appreciation and continuous dividends $q$, after discounting at the risk free rate $r$ is a martingale, see \cite{Shreve:1992}.

Taking first a derivative of $\hat{S}_t$
\begin{align} \label{mart}
d\hat{S}_t = d \left( e^{-r t} S_t e^{q t} \right) &=  e^{(q-r)t} \left[ (q-r) S_t dt + d S_t\right], \end{align}
\noindent and then an expectation of both parts we obtain
\begin{align} \label{martCond}
\EQ[d \left( e^{(q-r)t} S_t\right)] &=
e^{(q-r)t} \left\{ (q-r)  \EQ[S_t] d t + d \EQ[S_t] \right\}.
\end{align}
So in order for $\hat{S}_t$ to be a martingale, the RHS of \eqref{martCond} should vanish.
Solving the equation
\[  (q-r)  y(t) d t + d y(t) = 0, \qquad y(t) = \EQ[S_t | S_s], \ s < t\]
\noindent we obtain
\begin{align} \label{inLeft}
y(t) &=  \EQ[S_t | S_s] = S_s e^{(r-q) (t-s)}, \\
\EQ[d S_t  | S_s] &= d \EQ[S_t | S_s] = S_s (r-q) e^{(r-q) (t-s)} . \nonumber
\end{align}
On the other hand, from \eqref{sub}
\begin{align} \label{e1}
\EQ[d S_t  | S_s] &= \EQ[d \dgxt  | S_s] = \mu \EQ [\dgxt d\gxt  | S_s] + \EQ[\sigma(\dgxt) \dgxt d W_{\gxt}  | S_s] \\
&= \mu \EQ [\dgxt d \gxt  | S_s], \nonumber
\end{align}
\noindent because the process $W_\gxt$ is a local martingale, see \cite{RevuzYor1999}, chapter 6. Accordingly, the process $W_\gxt$ inherits this property from $W_{\Gamma_t}$, hence $\EQ[\sigma(\dgxt) \dgxt d W_\gxt] = 0$.

To proceed, assume the Gamma process $\Gamma_t$ is independent of $W_t$ (and, accordingly, $\gxt$ is independent of $W_\gxt$. Then the expectation in the RHS of \eqref{e1} can be computed, by first conditioning on $\gxt$, and then integrating over the distribution of $\gxt$ which can be obtained from \eqref{gammaDen} by replacing $t$ with $X(t)$, i.e.
\begin{align} \label{e2}
\EQ [\dgxt d \gxt  | S_s] &= \int_0^\infty  \EQ [\dgxt d \gxt| \gxt = \nu] \dfrac{\nu^{m-1} e^{-\nu m /X(t)}}{(t^*)^m \Gamma(m)} \\
&= \int_0^\infty  \EQ [D_\nu] \dfrac{\nu^{m-1} e^{-\nu m /X(t)}}{(t^*)^m \Gamma(m)} d \nu, \quad
\nu > 0, \quad m \equiv X(t)/t^*. \nonumber
\end{align}
The find $\EQ [D_\nu]$ we take into account \eqref{D} to obtain
\begin{align} \label{e3}
d \EQ [D_\nu] =  \EQ [d D_\nu]
= \EQ [\mu D_\nu d\nu + \sigma(D_\nu) D_\nu d W_\nu] = \mu \EQ [D_\nu] d\nu.
\end{align}
Solving this equation with respect to $y(\nu) = \EQ [D_\nu  | D_s]$,  we obtain $\EQ [D_\nu  | D_s] = D_s e^{\mu (\nu-s)}$. Since we condition on time $s$, it means that $D_s = D_{\Gamma_{X(s)}} = S_s$, and thus
$\EQ [D_\nu  | D_s] = S_s e^{\mu (\nu-s)}$.

Further, we substitute this into \eqref{e2}, set the parameter of the Gamma distribution $t^*$ to be $t^* = X(t)$ (so $m = 1$) and integrate to obtain
\begin{equation} \label{intRight}
\EQ[d S_t  | S_s] = \mu \EQ [\dgxt d \gxt] = S_s e^{- s \mu} \frac{\mu}{1 - \mu X(t)}.
\end{equation}
Finally, equating representations of $\EQ[d S_t  | S_s]$ obtained in \eqref{inLeft} and \eqref{intRight} we arrive at the equation for $X(t)$
\begin{equation} \label{eqX}
S_0 (r-q)e^{(r-q) (t-s)} = S_s e^{- s \mu} \frac{\mu}{1 - \mu X(t)}.
\end{equation}
Assuming $\mu = r - q$, this equation can be solved to provide
\begin{equation} \label{X}
X(t) = \dfrac{1 - e^{-(r-q) t}}{r-q}.
\end{equation}
This expression for $X(t)$ was also used in \cite{ELVG} for the ELVG. We already mentioned that the ELVG could be considered as an {\it arithmetic analog} of our model in this paper, which is {\it geometric} in $D_t$.

It is clear that in the limit $ r \to 0, \ q \to 0$ we have $X(t) = t$. Also based on \eqref{gammaExp}
\begin{equation} \label{gammaXexp}
\EQ[\Gamma_{X(t)}] = X(t).
\end{equation}
Function $X(t)$ starts at zero, i.e. $X(0) = 0$ \footnote{So our assumption made in above that $X(0) = 0$ is consistent.}, and is a continuous non-decreasing function of time $t$.
In more detail, if $r - q > 0$, function $X(t)$ is increasing in $t$ in all points  except at $t \to \infty$, where it tends to constant. However, the infinite time horizon doesn't have much practical sense, therefore for any finite time $t$ function $X(t)$ can be treated as an increasing function in $t$. In the other case when $r - q < 0$, function $X(t)$ is strictly increasing $\forall t \in [0,\infty)$. This means that, overall, $X(t)$ has all properties of a good clock. Accordingly, $\Gamma_{X(t)}$ has all properties of a random time.

Thus, we managed to demonstrate that with this choice of $\mu$ and $X(t)$ the right hands part of \eqref{martCond} vanishes, and our discounted stock process with allowance for non-zero interest rates and continuous dividends becomes a martingale. So the proposed construction can be used for option pricing.

This setting can be easily generalized for time-dependent interest rates $r(t)$ and continuous dividends $q(t)$. We leave it for the reader.

The next step is to establish a connection between the original and time-changed processes.
It is known from \cite{BochnerPDE1949}  that the process $G_{\Gamma_t}$ defined as
\[ d G_t = \sigma^2(G) G_t d W_t \]
\noindent is a time-homogeneous Markov process. Same is true for the process $(r-q) G_t dt$. Thus, the entire process $D_t$ defined in \eqref{D} is also a time-homogeneous Markov process. Accordingly, the semigroups $T^S_t$ of $S_t$ and $T^D_t$ of $D_{\Gamma_{X(t)}}$ are connected by the Bochner integral\footnote{Here it represents an expectation of the option price with respect to the second stochastic driver - stochastic clock $\nu$.}
\begin{equation} \label{BI}
\calT^S_t U(S) = \int_0^\infty \calT^D_\nu U(S) \mathbb{Q}\{\gxt \in d\nu\}, \quad \forall t \ge 0,
\end{equation}
\noindent where $U(S)$ is a function in the domain of both $\calT^D_t$ and $\calT^S_t$.
It can be derived by exploiting the time homogeneity of the $D$ process, conditioning on the gamma time first, and taking into account the independence of $\Gamma_t$ and $W_t$ (or $\Gamma_\gxt$ and $W_\gxt$ in our case).

As we set parameter $t^*$ of the gamma clock to $t^* = X(t)$, \eqref{BI} and \eqref{gammaDen} imply
\begin{equation} \label{BI1}
\calT^S_{t} U(S) = \int_0^\infty \calT^D_\nu U(S) \dfrac{e^{-\nu/X(t)}}{X(t)} d\nu.
\end{equation}
In what follows for the sake of brevity we call this model as the Geometric Local Variance Gamma model, or the GLVG.

\section{Forward equation for option prices} \label{ForwardEq}

In this section we derive a forward equation for put option prices, which is an analog of the Dupire equation for the standard local volatility model.
In doing so, we closely follow the description in the corresponding section of \cite{ELVG}, as from the derivation point of view the GLVG differs from the ELVG just by the definitions of infinitesimal generator $\calA$ of the process $D_t$.

Let us interpret the index $t$ of the semigroup $\calT^S_t$ as the maturity date $T$ of an European claim with the valuation time $t_v = 0$. Also let the test function $U(S)$ be the payoff of this European claim, i.e.
\begin{equation} \label{payoff}
U(S_T) = e^{-r T}(K - S_T)^+.
\end{equation}
Then define
\begin{equation} \label{P0}
P(S_0,T,K) = \calT^S_T U(S_0)
\end{equation}
\noindent as the European Put value with maturity $T$ at time $t=0$ in the LVGE model. Similarly
\begin{equation} \label{P0D}
P^D(S_0,\nu,K) = \calT^D_\nu U(S_0)
\end{equation}
\noindent would be the European Put value with maturity $\nu$ at time $t=0$ in the model of \eqref{D}\footnote{Below for simplicity of notation we drop the subscript '0'  in $S_0$.}. Then the Bochner integral in \eqref{BI1} takes the form
\begin{equation} \label{Bochner2}
P(S,T,K) =  \int_0^\infty P^D(S,\nu,K) p e^{- p \nu} d \nu,  \quad p \equiv 1/X(T).
\end{equation}
Thus, $P(S,T,K)$ is represented by a Laplace-Carson transform of $P^D(S,\nu,K)$ with $p$ being a parameter of the transform. Note that
\begin{equation} \label{init}
P(S,0,K) = P^D(S,0,K) = U(S).
\end{equation}
To proceed, we need an analog of the Dupire forward PDE for $P^D(S,\nu,K)$.

\subsection{Dupire-like forward PDE \label{dupFWPDE}}

Despite this can be done in many different ways, below for the sake of compatibility we do it in the spirit of \cite{CarrNadtochiy2017}.

First, differentiating \eqref{P0D} by $\nu$ with allowance for \eqref{semig} yields
\begin{align} \label{Dt}
\nabla_\nu P^D(S,\nu,K) &= e^{-r \nu} e^{\nu \calA}\left[ \calA - r\right] U(S)
= e^{-r \nu} \EQ \left[ \calA - r\right] U(S).
\end{align}
We take into account the definition of the generator $\calA$ in \eqref{gen}, and also remind that at $t=0$ we have $D_0 = S_0 \equiv S$. Then \eqref{Dt} transforms to
\begin{align} \label{Dt1}
\nabla_\nu P^D(S,\nu,K) = &-r P^D(S,\nu,K) + \left(r-q\right) S \nabla_S P^D(S,\nu,K)
+ e^{-r \nu}\frac{1}{2} \EQ \left[ \sigma^2(S) S^2 \nabla_S^2 U(S) \right].
\end{align}
However, we need to express the forward equation using a pair of independent variables $(\nu,K)$ while \eqref{Dt} is derived in terms of $(\nu,S)$. To do this, observe that
\begin{align} \label{sir}
\EQ \left[\sigma^2(S) S^2 \nabla_S^2 U(S) \right] &= \EQ \left[\sigma^2(S)S^2  \delta(K-S)\right] = \EQ \left[\sigma^2(K) K^2 \delta(K-S)\right] \\
&=  \EQ \left[\sigma^2(K) K^2 \nabla_K^2 U(S)\right] = e^{r \nu} \sigma^2(K) \nabla_K^2 P^D(S,\nu,K). \nonumber
\end{align}
\noindent where the sifting property of the Dirac delta function $\delta(S-K)$ has been used. Also
\begin{align} \label{term2}
-r & P^D(S,\nu,K) + (r-q) S \nabla_S P^D(S,\nu,K) \\
&= e^{-r \nu} \EQ\left[ -r (K-S)^+ + (r-q)S \fp{(K-S)^+}{S} \right] \nonumber \\
&= e^{-r \nu} \EQ\left[ -r (K-S)^+ - (r-q)(K-S) \fp{(K-S)^+}{S} + (r-q)K \fp{(K-S)^+}{S}\right] \nonumber \\
&= e^{-r \nu} \EQ\left[ -r (K-S)^+ + (r-q)(K-S)^+  - (r-q)K \fp{(K-S)^+}{K}\right] \nonumber \\
&= -q P^D(S,\nu,K) - (r-q) K \nabla_K P^D(S,\nu,K). \nonumber
\end{align}
Therefore, using \eqref{sir} and \eqref{term2}, \eqref{Dt} could be transformed to
\begin{align} \label{Dup1}
\nabla_\nu P^D(S,\nu,K) &= -q P^D(S,\nu,K) - (r-q) K \nabla_K P^D(S,\nu,K) + \frac{1}{2} \sigma^2(K) K^2 \nabla_K^2 P^D(S,\nu,K) \nonumber  \\
&\equiv {\calA}^K P^D(S,\nu,K),  \\
{\calA}^K &= -q  - (r-q) K \nabla_K  + \frac{1}{2} \sigma^2(K) K^2 \nabla_K^2. \nonumber
\end{align}
This equation looks exactly like the Dupire equation with non-zero interest rates and continuous dividends, see, e.g., \cite{Tysk2012} and references therein. Note, that $\calA^K$ is also a time-homogeneous generator.

\subsection{PDDE for a single term} \label{FPDDder}

Our final step is to apply the linear differential operator $\cal L$ defined in  \eqref{Dup1} to both parts of \eqref{Bochner2}. Using time-homogeneity of $D_t$ and again the Dupire equation \eqref{Dup1}, we obtain
\begin{align} \label{b2}
-q &P(S,T,K) - (r-q) K \nabla_K P(S,T,K) + \frac{1}{2} \sigma^2(K) K^2 \nabla_K^2
P(S,T,K) \\
&= \int_0^\infty p e^{- p \nu} \left[-q P^D(S,\nu,K) - (r-q) K \nabla_K P^D(S,\nu,K) + \frac{1}{2} \sigma^2(K) K^2 \nabla_K^2 P^D(S,\nu,K)\right] d \nu \nonumber \\
&= \int_0^\infty p e^{- p \nu} \nabla_\nu P^D(S,\nu,K) d \nu  = - p P^D(S,0,K) + p \int_0^\infty P^D(S,\nu,K) p e^{- p \nu} d \nu \nonumber \\
&= p\left[P(S,T,K) - P^D(S,0,K) \right] = p\left[P(S,T,K) - P(S,0,K) \right], \nonumber
\end{align}
\noindent where in the last line we took into account \eqref{init}.

Thus, finally $P(S,T,K)$ solves the following problem
\begin{align} \label{finDup}
-q P(S,T,K) - (r-q) K \nabla_K P(S,T,K) &+ \frac{1}{2} \sigma^2(K) K^2 \nabla_K^2
P(S,T,K) = \dfrac{P(S,T,K) - P(S,0,K)}{X(T)}, \nonumber \\
P(S,0,K) &= (K-S)^+.
\end{align}
In contrast to the Dupire equation which belongs to the class of PDE, \eqref{finDup} is an ODE, or more precisely a partial divided-difference equation (PDDE), since the derivative in time in the right hands part is now replaced by a divided difference. In the form of an ODE it reads
\begin{equation} \label{finDupPut}
\left[\frac{1}{2} \sigma^2(K) K^2 \nabla_K^2 - (r-q) K \nabla_K - \left(q + \dfrac{1}{X(T)}\right) \right] P(S,T,K) =
- \dfrac{P(S,0,K)}{X(T)}.
\end{equation}
This equation could be solved analytically for some particular form of the local volatility function $\sigma(K)$ which are considered in the next Section. Also in the same way a similar equation could be derived for the Call option price $C_0(S,T,K)$ which reads
\begin{align} \label{finDupCall}
\Big[\frac{1}{2} \sigma^2(K) K^2 \nabla_K^2 + (r-q) K \nabla_K &- \left(q + \dfrac{1}{X(T)}\right) \Big] C_0(S,T,K) = - \dfrac{C_0(S,0,K)}{X(T)}, \nonumber \\
C_0(S,0,K) &= (S-K)^+.
\end{align}

Solving \eqref{finDupPut} or \eqref{finDupCall} provides the way to determine $\sigma(K)$ given market quotes of Call and Put options with maturity $T$. However, this allows calibration of just a single term. Calibration of the entire local volatility surface, in principle, could be done term-by-term (because of the time-homogeneity assumption) if \eqref{finDupPut}, \eqref{finDupCall} could be generalized to this case.

\subsection{PDDE for multiple terms}

This generalization can be done in the same way as presented in \cite{ELVG}, Section 4. Therefore, we refer the reader to that Section while here provide just some useful comments.

To address calibration of multiple smiles  we need to relax the assumption about time-homogeneity of the $D_t$ process defined in \eqref{D}. We assume that the local variance $\sigma(D_t)$ is no more time-homogeneous, but a piecewise constant function of time $\sigma(D_t,t)$.

Let $T_1, T_2, \ldots, T_M$ be the time points at which the variance rate $\sigma^2(D_t)$ jumps deterministically. In other words, at the interval $t \in [T_0,T_1)$, the variance rate is $\sigma^2_0(D_t)$, at $t \in [T_1, T_2)$ it is $\sigma^2_1(D_t)$, etc. This can be also represented as
\begin{align} \label{sigmaPW}
\sigma^2(D_t,t) &= \sum_{i=0}^M \sigma^2_i(D_t) w_i(t), \\
w_i(t) &\equiv {\mathbf 1}_{t - T_i} - {\mathbf 1}_{t - T_{i+1}}, \ i=0,\ldots,M,
\quad T_0 = 0, \ T_{M+1} = \infty. \nonumber
\end{align}
Note, that
\[ \sum_{i=0}^M w_i(t) = {\mathbf 1}_t - {\mathbf 1}_{t-\infty} = 1, \quad \forall t \ge 0.\]
\noindent Therefore, in case when all $\sigma^2_i(D_t)$ are equal, ie, independent on index $i$, \eqref{sigmaPW} reduces to the case considered in the previous Sections.

This implies that the volatility $\sigma(D_t)$ jumps as a function of time at the calendar times $T_0, T_1,\ldots,T_M$, and not at the business times $\nu$ determined by the Gamma clock. Otherwise, the volatility function would have been changed at random (business) times which means it is stochastic. But this definitely lies out of scope of our model. Therefore, we need to change \eqref{sigmaPW} to
\begin{align} \label{sigmaPW_exp}
\sigma^2(D_t,t) &= \sum_{i=0}^M \sigma^2_i(D) {\bar w}_i(\EQ(t)), \\
{\bar w}_i(\EQ(t)) &= {\mathbf 1}_{X^{-1}(t - T_i)} - {\mathbf 1}_{X^{-1}(t - T_{i+1})}, \ i=0,\ldots,M, \nonumber \\
X^{-1}(t) &= \dfrac{1}{q-r} \log \left[1 - (r-q)t \right].
\end{align}
As per the last line, $X(t)$ exists $\forall t \ge 0$ if $q > r$, and $\forall t < 1/(r-q)$ if $r > q$.

Hence, when using \eqref{sub} we have
\begin{align} \label{sigmaPW_exp1}
\sigma^2(D_t, t)\Big|_{t = \Gamma_{X(t)}} &= \sum_{i=0}^M \sigma^2_i(D) \bar{w}_i(X(t)) = \sum_{i=0}^M \sigma^2_i(D) w_i(t).
\end{align}

Accordingly, if the calendar time $t$ belongs to the interval $T_0 \le t < T_1$, the infinitesimal generator $\calA$ of the semigroup $\calT^D_\nu$ is a function of $\sigma(D_t)$ (and not on $\sigma(D_\nu)$). As at $T_0 \le t < T_1$ we assume $\sigma(D) = \sigma_0(D)$, i.e. is constant in time, it doesn't depend of $\nu$. Thus, $\calA$ (which for this interval of time we will denote as $\calA_0)$ is still time-homogeneous.

Similarly, one can see, that for $T_1 \le t < T_2$ the infinitesimal generator $\calA_1$ of the semigroup $\calT^D_\nu$ is also time-homogeneous and depends on $\sigma_1(D)$, etc.

Further, similar to \cite{ELVG} it could be shown that the forward partial divided difference equation for the Put price $P(S,T_i,K), \ i=1,\ldots,M$ reads
\begin{equation} \label{finDupPutM}
\left[\frac{1}{2} \sigma^2(K) K^2 \nabla_K^2 - (r-q) K \nabla_K - \left(q + \dfrac{1}{X(T_i) - X(T_{i-1}}\right) \right] P(S,T_i,K) =
- \dfrac{P(S,T_{i-1},K)}{X(T_i) - X(T_{i-1})}.
\end{equation}
Here the local variance function $\sigma^2(K) =  \sigma^2_i(K)$ as it corresponds to the interval $(T_{i-1}, T_i]$ where the above ODE is solved.

\eqref{finDupPutM} is a recurrent equation that can be solved for all $i=1,\ldots,M$ sequentially starting with $i = 1$ subject to some boundary conditions.

\subsection{Boundary conditions}

In many financial models where dynamics of the stock price is represented by a geometric Brownian motion (perhaps with local or stochastic volatility), for instance, the celebrating Black-Scholes model, the boundary  condition at $K \to \infty$ is set to be
\[ P(S,T_i,K) \to \disD_i K - \disQ_i S, \quad K \to \infty, \]
\noindent where $\disD_i = e^{-r T_i}$ is the discount factor, and $\calQ_i = e^{-q T_i}$. Indeed, as it could be easily checked this condition is a valid solution of the Dupire forward equation \eqref{Dup1}, and also reflects the fact that at $K \to \infty$ the Put option price should be linear in $K$. However, this boundary condition doesn't solve \eqref{finDupPut}, so it could not be used in our model.

Therefore, we propose to setup the boundary condition at $K \to \infty$ by still assuming it to be a linear function of $K$ of the form
\begin{equation} \label{bcGen}
\lim_{K \to \infty} P(S,T,K) = A(T) K - B(T) S,
\end{equation}
\noindent where $A(T), B(T)$ are some functions of maturity $T$ to be determined, so the expression in \eqref{bcGen} solves \eqref{finDupPut}.

Obviously, $T_0=0$ implies $A(T_0) = B(T_0) = 1$. Then we can proceed recursively. For the next given maturity $T = T_1$ plugging  in \eqref{bcGen} into \eqref{finDupPutM} we obtain at $K \to \infty$
\begin{align}
-(r-q)K A(T_1)p_1 &-(p_1 q + 1)( A(T_1) K - B(T_1) S) = -P(S,T_0,K), \\
P(S,T_0,K) &= A(T_0) K - B(T_0) S = K - S, \nonumber \\
p_j &= X(T_j) - X(T_{j-1}) > 0. \nonumber
\end{align}
From these equations we obtain
\begin{equation} \label{rec1}
B(T_1) = \frac{1}{p_1 q +1}, \qquad A(T_1) = \frac{1}{p_1 r +1}.
\end{equation}
So in this case $A(T_1), B(T_1)$ are an analog of some kind of discrete compounding.

Proceeding recursively, we derive a general relationship
\begin{align} \label{recGen}
B(T_i) &= \frac{B(T_{i-1})}{p_i q + 1} = \frac{1}{\prod_{k=1}^i (p_i q + 1)}, \\
A(T_i) &= \frac{A(T_{i-1})}{p_i r +1} = \frac{1}{\prod_{k=1}^i (p_i r + 1)}, \qquad i=1,\ldots,M. \nonumber
\end{align}

Therefore, in our model the natural boundary conditions for the Put option price are
\begin{equation} \label{bcP}
\begin{cases}
P(S,T_i,K) = 0, & K \to 0, \\
P(S,T_i,K) = A(T_i) K - B(T_i) S \approx A(T_i) K, & K \to \infty, \\
\end{cases}
\end{equation}
\noindent

A similar equation can be obtained for the Call option prices, which reads
\begin{equation} \label{finDupCallM}
\left[\frac{1}{2} \sigma^2(K) K^2 \nabla_K^2 + (r-q) K \nabla_K - \left(q + \dfrac{1}{X(T_i) - X(T_{i-1}}\right) \right] C(S,T_i,K) =
- \dfrac{C(S,T_{i-1},K)}{X(T_i) - X(T_{i-1})}.
\end{equation}
\noindent subject to the boundary conditions
\begin{equation} \label{bcC}
\begin{cases}
C(S,T_i,K) = B(T_i) S, & K \to 0, \\
C(S,T_i,K) = 0, & K \to \infty. \\
\end{cases}
\end{equation}

\section{Piecewise models of local variance/volatility} \label{cases}

To calibrate the local volatility surface by solving \eqref{finDupPutM} we need to make further assumptions about the shape of the local volatility surface. To recall, we assume this surface to be piecewise constant in time. In the strike space \cite{CarrNadtochiy2017} considered it to be a piecewise constant, while in \cite{ELVG} a piecewise linear local variance in the strike space was considered. As shown in \cite{ELVG} in that cases \eqref{finDupPutM} can be solved in closed form.

In this paper we want to extend a class of local volatility models that allow a closed form solution.  To proceed, we start by doing a change of the dependent variable from $P(S, T_j, K)$ to
\begin{equation} \label{P2V}
V (S, T_j, K) = P(S, T_j, K) - [A(T_j)K - B(T_j)S]^+,
\end{equation}
\noindent where $V$ is known as a {\it covered} Put. This definition of $V$ allows re-writing \eqref{finDupPutM} in a more elegant form
\begin{align} \label{ode1}
&- v_j(x) x^2 V_{x,x}(x) + b_{1,j} x V_x(x) + b_{0,j} V(x) = c_j(x), \\
b_{1,j} &= p_j(r - q), \quad b_{0,j} = p_j q + 1,
\quad c_j(x) = V(S,T_{j-1}, x), \quad v_j(x) = p_j \sigma^2(x)/2, \nonumber
\end{align}
\noindent where $V(x) = V(S, T_j, x)$ and $x= K/S$ is the inverse moneyness.

Accordingly, based on the definition of $V(x)$ and \eqref{bcP}, the boundary conditions to \eqref{ode1} become homogeneous
\begin{equation} \label{bc}
\begin{cases}
V(x) = 0, & x \to 0, \\
V(x) = 0, & x \to \infty. \\
\end{cases}
\end{equation}

In the next sections we consider several popular approximations of the local volatility surface in the strike space. Each approximation assumes some functional form of the local volatility curve in the strike space, which is a strip of the volatility surface given time to maturity $T$. Thus, parameters of these approximations change with time. Also further on for the sake of certainty we assume that $r > q > 0$, but this assumption could be easily relaxed.

\subsection{Local variance piecewise linear in a log-strike space} \label{locVarLog}

Suppose that for each maturity $T_j, \ j \in [1, M]$ the market quotes are provided for a set of strikes $K_i, i = 1,\ldots,n_j$ where these strikes are assumed to be sorted in the increasing order. Then the corresponding continuous piecewise linear local variance function $\sigma^2_j(\chi)$ at the interval $[\chi_i, \chi_{i+1}], \ \chi = \log K_i/S,$ reads
\begin{equation} \label{linSigChi}
v_{j,i}(\chi) = v^0_{j,i} + v^1_{j,i} \chi.
\end{equation}
\noindent Here we use the super-index 0 to denote a level $v^0$ , and the super-index 1 to denote a slope $v^1$. Subindex $i = 0$ in $v^0_{j,0}, v^1_{j,0}$ corresponds to the interval $(0, \chi_1]$. Since $v_j(\chi)$ is a continuous function in $\chi$, we have
\begin{equation} \label{contChi}
v^0_{j,i} + v^1_{j,i} \chi_{i+1} = v^0_{j,i+1} + v^1_{j,i+1} \chi_{i+1}, \qquad i=0,\ldots,n_j-1.
\end{equation}
This means that the first derivative of $v_j(\chi)$ experiences a jump at points $\chi_i, \ i \in Z \cap [1, n_j]$. As we assumed that $v(\chi, T)$ is a piecewise constant function of time, $v^0_{j,i}, v^1_{j,i}$ do not depend on $T$ at the intervals $[T_j, T_{j+1}), \ j \in [0, M-1]$, and jump to the new values at the points $T_j, \ j \in Z \cap [1, M]$.

A simple analysis shows that under this assumption by making a change of variables $x \mapsto \chi$, \eqref{ode1} could be transformed to
\begin{equation} \label{odeLog}
- v(\chi) V_{\chi,\chi}(\chi) + (b_{1} + v(\chi)) V_\chi(\chi) + b_{0} V(\chi) = c(\chi),
\end{equation}
\noindent where for simplicity of notation we dropped index $j$.

This equation has the same type as that considered in \cite{ItkinLipton2017}, Section 2, and its solution could also be expressed in terms of confluent Hypergeometric functions, see \cite{PolyaninSaitsevODE2003}
\begin{align} \label{solInhom}
V(\chi) &= C_1 y_1(\chi) + C_2 y_2(\chi) + I_{12}(\chi) \\
I_{12}(\chi) &= y_2(\chi) \int \dfrac{ y_1(\chi) c(\chi)}{(b_2 + a_2 \chi)W}d \chi - y_1(\chi) \int  \dfrac{ y_2 c(\chi)}{(b_2 + a_2 \chi)W}d \chi, \nonumber
\end{align}
\noindent where $W = y_1 (y_2)_\chi - y_2 (y_1)_\chi$ is the so-called Wronskian of the fundamental solutions $y_1,y_2$. Thus, the problem is reduced to finding suitable fundamental solutions of the homogeneous version of \eqref{solInhom}. Based on \cite{PolyaninSaitsevODE2003}, if $a_2 \ne 0$ and $a_0 \ne 0$, the general solution reads
\begin{align} \label{homog}
V(\chi) &= (a_2 z)^{\beta_1-1} \mathcal{J}\left(\alpha_1, \beta_1, z\right), \\
z &= \chi + \dfrac{b_2}{a_2}, \quad \alpha_1 = 1 + \dfrac{b_0 + b_1}{a_2}, \quad \beta_1 = 2 + \dfrac{b_1}{a_2}. \nonumber
\end{align}
Here $\mathcal{J}(a,b,z)$ is an arbitrary solution of the degenerate Hypergeometric equation, i.e., Kummer'’s function, \cite{as64}. Two types of Kummer'’s functions are known, namely $M(a,b,z)$ and $U(a,b,z)$, which are Kummer’s functions of the first and second kind \footnote{Due to the linearity of the degenerate Hypergeometric equation any linear combination of Kummer's functions also solves this equation.}.

Accordingly, the approach of \cite{ItkinLipton2017} can be directly applied to obtain a closed form solution of \eqref{solInhom}. In particular, in the vicinity of the origin the numerically satisfactory pair is, \cite{Olver1997}
\begin{align} \label{Kummer0}
y_1(\chi) &= (a_2 z)^{\beta_1-1} M\left(\alpha_1, \beta_1, z\right), \\
y_2(\chi) &= (a_2)^{\beta_1-1} M\left(\alpha_1 - \beta_1 +1, 2-\beta_1, z\right). \nonumber \\
W &= a_2^{2 \beta_1-2} e^{z} z^{\beta_1-2} \sin(\pi \beta_1)/\pi. \nonumber
\end{align}
However, in the vicinity of infinity the numerically satisfactory pair is, \cite{Olver1997}
\begin{align} \label{Kummer1}
y_1(\chi) &=   (a_2 z)^{\beta-1} U\left(\alpha_1, \beta_1, z\right), \\
y_2(\chi) &= e^z (a_2 z)^{\beta-1} U\left(\beta_1 - \alpha_1, \beta_1, -z\right). \nonumber \\
W &= (-1)^{\alpha_1 - \beta_1} a_2^{2 \beta_1-2} e^{z} z^{\beta_1-2}. \nonumber
\end{align}

\subsection{Local variance piecewise linear in the strike space} \label{locvar}

Another tractable model is where the local variance is piecewise linear in the strike space. In particular, this is the model we used in \cite{ELVG}.

Similar to the previous section, the corresponding continuous piecewise linear local variance function $v_j(x)$ at the interval $[x_i, x_{i+1}]$ reads
\begin{equation} \label{linSig}
v_{j,i}(x) = v^0_{j,i} + v^1_{j,i} x,
\end{equation}
\noindent where, however, it is now a function of $x$ rather than $\chi$. Since
$v_j(x)$ is a continuous function in $x$, we have
\begin{equation} \label{cont1}
v^0_{j,i} + v^1_{j,i} x_{i+1} = v^0_{j,i+1} + v^1_{j,i+1} x_{i+1}, \qquad i=0,\ldots,n_j-1.
\end{equation}
This means that the first derivative of $v_j(x)$ experiences a jump at points $x_i, \ i \in Z \cap [1, n_j]$. As we assumed that $v(x, T)$ is a piecewise constant function of time, $v^0_{j,i}, v^1_{j,i}$ don't depend on $T$ at the intervals $[T_j, T_{j+1}), \ j \in 0, M-1]$, and jump to the new values at the points $T_j, \ j \in Z \cap [1, M]$.

The \eqref{ode1} can be solved by induction. One starts with $T_0 = 0$, and at each time interval $[T_{j-1}, T_j], \ j \in Z \cap [1, M]$ solves the problem \eqref{ode1} for $V(x)$, and then obtains $P(S, T_j,x)$ from \eqref{P2V}. Accordingly, the solution of \eqref{ode1} can be constructed separately for each interval $[x_{i-1}, x_i]$.

Substituting the representation \eqref{linSig} into \eqref{ode1}, for the $i$-th spatial interval we obtain
\begin{align} \label{ode2}
- (b_2 + a_2 x) x^2 V_{x,x}(x) &+ b_1 x V_x(x) + b_{0,j} V(x) = c(x), \\
b_2 &= v_{j,i}^0, \quad a_2 = v_{j,i}^1. \nonumber
\end{align}
Again, \eqref{ode2} is an {\it inhomogeneous} ordinary differential equation, and its solution can be represented in the form of \eqref{solInhom} with
\begin{equation} \label{I12v}
I_{12}(x) = -y_2(x) \int \dfrac{ y_1(x) c(x)} {(b_2 + a_2 x)x^2 W(x)}d x + y_1(x) \int  \dfrac{ y_2(x) c(x)}{(b_2 + a_2 x) x^2 W(x)}d x \equiv J_1 + J_2.
\end{equation}

The corresponding homogeneous equation can be solved as follows. First, if $b_2 \ne 0$ we make a change of independent variable $x \mapsto z = - a_2 x/b_2$. As the result the homogeneous \eqref{ode2} takes the form
\begin{equation} \label{ode2hom}
b_2 (z-1) z V_{z,z}(z) + b_1 z V_z(z) + b_{0} V(z) = 0. \\
\end{equation}
Then we make a change of the dependent variable $V(z) \mapsto z^m G(z)$ with $m$ being some constant for the given time slice. This leads to the equation
\begin{align} \label{ode2hom2}
 z^m [ \gamma &+ b_2 (m-1) m z ] G(z) + z^{m+1}[b_1 + 2 b_2 m (z-1)] G'(z)  + b_2 (z-1) z^{m+2} G''(z) = 0, \\
 \gamma &= b_0 + m(b_2 + b_1 - b_2 m). \nonumber
\end{align}
Next we solve for $m$ which makes $\gamma$ vanishing, to obtain
\begin{equation} \label{m}
m^{\pm} = \dfrac{b_2 + b_1 \pm \sqrt{4 b_2 b_0 + (b_2 + b_1)^2}}{2 b_2.}
\end{equation}
It is worth mentioning that if the determinant  $D$ in this expression is negative, both $m^+, \ m^-$ become complex.  However, this is not a problem for the solution as coefficients $C_1,C_2$ in \eqref{solInhom} could be complex as well, and such that the Put price is real.

Substituting this into \eqref{ode2hom2} and rearranging we obtain
\begin{equation} \label{hyper}
-m(m-1)G(z) + \left(2 m - \frac{b_1}{b_2} - 2m z\right) G'(z) + z(1-z)G''(z) = 0, \quad m \in [m^+, m^-],
\end{equation}
\noindent which is a Hypergeometric equation. As $m$ can take two values, we need to choose the right one such that the final solution would obey the boundary conditions.

Combining all the above steps together, the solution of \eqref{ode2hom} could be written as
\begin{align} \label{homog2}
y_1(x) &= z^{m} \left[\pFq{2}{1}{m-1, m, c; z} \right], \\
y_2(x) &= z^{m} \left[z^{1-c} \pFq{2}{1}{m - c, m +1 - c, 2 - c; z} \right],
\nonumber \\
m &= m^+, \quad c =  2m - \frac{b_1}{b_2}, \quad z = - \frac{a_2}{b_2} x, \nonumber
\end{align}
Here $\pFq{2}{1}{a, b, c; z}$ is the ordinary Hypergeometric function, \cite{Olver1997}.  It has regular singularities at $z=0, 1, \infty$. In terms of the solution in \eqref{homog}, these  singularities correspond to $K = 0, \ v = 0$ and $K \to \infty$. We will show below that at $K \to \infty$ the coefficient $a_2$ for this interval is usually positive, so the variance is positive. However, the sign of $b_2$ could be both plus and minus. Therefore,
if $b_2 > 0$ at this interval, we have $x \to \infty, \ z \to -\infty$. If $b_2 < 0$ at this interval, we have $x \to \infty, \ z \to \infty$.

When none of $c, c -a -b, a - b$ is an integer, we have a pair of fundamental solutions $f_1(x), f_2(x)$ that in \eqref{homog} are represented by expressions in square brackets. It is known that this pair is numerically satisfactory, \cite{Olver1997} aside of singularities at $z = 1$ and $z \to \infty$. Wronskian of these fundamental solutions $W(f_1(x), f_2(x))$ is
\begin{equation*}
W(f_1(x), f_2(x)) = (1-c) z^{-c}(1-z)^{c-2m}, \quad z = - a_2 x/b_2.
\end{equation*}
Accordingly,
\begin{equation} \label{wron1}
W(y_1(x), y_2(x)) = - \frac{a_2(1-c)}{b_2}z^{2m - c}(1-z)^{c-2m}, \quad z = - a_2 x/b_2.
\end{equation}
In the vicinity of {\bf singularity at $\boldsymbol{z = 1}$} this pair, however, is not numerically satisfactory. Then we have to use another solution of \eqref{hyper} which is, \cite{Olver1997}
\begin{align} \label{homog3}
y_1(x) &= z^{m} \left[\pFq{2}{1}{m-1, m, 2m-c; 1-z} \right], \\
y_2(x) &= z^{m} \left[(1-z)^{c - 2m + 1}\pFq{2}{1}{c-m+1, c-m, c - 2m +2; 1-z} \right], \nonumber \\
W(y_1(x), y_2(x)) &= -\frac{a_2(2m-1-c)}{b_2}(1-z)^{c-2m}z^{2m - c}, \quad z = - a_2 x/b_2. \nonumber
\end{align}
The numerically satisfactory fundamental solutions in the vicinity of {\bf singularity at $\boldsymbol{z = \infty}$} is described in \ref{App1}.

However, we cannot use this solution at $z \to \infty$ as well as to use the solution in \eqref{homog2} at $z \to 0$. This is caused by the Roger Lee's moment matching formula, \cite{Lee2004} which states that in the wings the implied variance surface should be at most linear in the normalized strike (or log-strike). It is also shown in  \cite{Friz2013,Friz2015}, that the asymptotic behavior of the local variance is linear in the log strike at both $K \to \infty$ and $K \to 0$. While the result for $K \to 0$ is shown to be true at least for the Heston and Stein-Stein models, the result for $K \to \infty$ directly follows from Lee's moment formula for the implied variance $v_I$ and the representation of $\sigma^2$ via the total implied variance $w = v_I T$, \cite{Lipton2001, Gatheral2006}
 \begin{equation} \label{lv}
w_L \equiv \sigma^2(T, K) T = \dfrac{T \dd_T w }{\left(1-\frac{X\dd_X w }{2 w }\right)^2
- \frac{(\dd_X w)^2}{4}\left(\frac{1}{ w }+\frac{1}{4}\right)+\frac{\dd^2_X w }{2}},
\end{equation}
\noindent where $w = w(X,T), X = \log K/F$ and $F = S e^{(r-q)T}$ is the stock forward price.

Thus, the considered model of the local variance linear in strike is not applicable at the first $0 \le x \le  x_1$ and the last $x_{n_j} < x < \infty$ strike intervals for every smile $T = T_j$ as it violets Lee's formula. Therefore, at these two intervals we use the model discussed in Section~\ref{locVarLog} where the local variance is linear in the log-strike.

It is interesting to mention, that in \cite{ItkinLipton2017, ELVG} and in section~\ref{locVarLog} the closed form solution was obtained in terms of Kummer's functions. Here the solution is expressed via Hypergeometric functions $\pFq{2}{1}{a,b,c;x}$.

As two solutions $y_1(x)$ and $y_2(x)$ are independent, \eqref{solInhom} is a general solution of \eqref{ode2}. Two constants $C_1,C_2$ should be determined based on the boundary conditions for the function $y(x)$.

The boundary conditions for the ODE \eqref{ode2} in the $x$ space at zero and infinity are given in \eqref{bc}, i.e. they are homogeneous. Based on the usual shape of the local variance curve and its positivity, for $x \to 0$, we expect that $v^1_{j,i} < 0$.  Similarly, for $x \to \infty$ we expect that $v^1_{j,i} > 0$. In between these two limits the local variance curve for a given maturity $T_j$ is assumed to be continuous, but the slope of the curve could be both positive and negative, see, e.g., \cite{ItkinSigmoid2015} and references therein.

\subsection{Local volatility piecewise linear in the strike space}

Another popular model is where the local volatility is assumed to be piecewise linear in the strike space. This model previously was frequently considered in the literature, e.g., \cite{HW2015,KC2017}. Below we show that with this assumption our model remains tractable, and a closed form solution can be obtained  by using the same approach as elaborated on in \cite{ItkinLipton2017, ELVG}.

Accordingly, the corresponding continuous piecewise linear local volatility function $\sigma_j(x)$ on the interval $[x_i, x_{i+1}]$ reads
\begin{equation} \label{linSig1}
\sigma_{j,i}(x) = \sigma^0_{j,i} + \sigma^1_{j,i} x,
\end{equation}
Since $\sigma_j(x)$ is a continuous function in $x$, we have
\begin{equation} \label{contSig}
\sigma^0_{j,i} + \sigma^1_{j,i} x_{i+1} = \sigma^0_{j,i+1} + \sigma^1_{j,i+1} x_{i+1}, \qquad i=0,\ldots,n_j-1.
\end{equation}
Again, this means that the first derivative of $\sigma_j(x)$ experiences a jump at points $x_i, \ i \in Z \cap [1, n_j]$. As $\sigma(x, T)$ is a piecewise constant function of time, $\sigma^0_{j,i}, \sigma^1_{j,i}$ do not depend on $T$ at the intervals $[T_j, T_{j+1}), \ j \in 0, M-1]$, and jump to the new values at the points $T_j, \ j \in Z \cap [1, M]$.

Substituting the representation \eqref{linSig1} into \eqref{ode1}, for the $i$-th spatial interval we obtain
\begin{align} \label{ode3}
- (b_2 + a_2 x)^2 x^2 V_{x,x}(x) &+ b_1 x V_x(x) + b_{0,j} V(x) = c(x), \\
b_2 &= \sigma_{j,i}^0, \quad a_2 = \sigma_{j,i}^1. \nonumber
\end{align}

Again, \eqref{ode3} is an {\it inhomogeneous} ordinary differential equation, and its solution can be represented in the form of \eqref{solInhom} with
\begin{equation*} \label{I12sigma}
I_{12}(x) = -y_2(x) \int \dfrac{ y_1(x) c(x)} {(b_2 + a_2 x)^2 x^2 W(x)}d x + y_1(x) \int  \dfrac{ y_2(x) c(x)}{(b_2 + a_2 x)^2 x^2 W(x)}d x \equiv L_1 + L_2.
\end{equation*}

The corresponding homogeneous equation can be solved as follows. First, if $b_2 \ne 0, b_2 + a_2 x \ne 0$ we make a change of independent variable $x \mapsto z = a_2 b_1 x/[b_2^2 (b_2 + a_2 x)]$. As the result the homogeneous \eqref{ode3} takes the form
\begin{equation*} \label{ode2homVol}
b_2 z^2 (-b_1 + b_2^2 z) V_{z,z}(z) +  z \left[2 b_2^4 z + \left(b_1 - b_2^2 z \right)^2\right] V_z(z) + b_0(b_1 - b_2^2 z) V(z) = 0.
\end{equation*}
Next we make a change of the dependent variable
\[ V(z) \mapsto z^{k_1} \left(\frac{z}{b_2^2 z + b_1}\right)^{k_2} G(z) \]
with $k_1, k_2$ being some constants for the given time slice. This leads to the equation
\begin{align} \label{G}
0 &= - b_2^2 z \left(b_1 - b_2^2 z \right)^2 G''(z) + f_1(z) G'(z) + f_0(z) G(z), \\
f_1(z) &= z \left(b_1 - b_2^2 z\right) \left[b_2^4 z (2 k_1 + z +2) - 2 b_2^2 b_1 (k_1 + k_2 + z) + b_1^2\right], \nonumber \\
f_0(z) &= q_0 + q_1 z + q_2 z^2 - b_2^6 k_1 z^3, \nonumber \\
q_2 &=  b_2^4 \left[b_0 - b_2^2 k_1 (k_1+1) + b_1 (3 k_1 + k_2)\right], \nonumber \\
q_1 &= b_2^2 {b_1} \left[2 b_2^2 {k_1} ({k_1}+{k_2})-2 {b_0}-{b_1} (3 {k_1}+2 {k_2})\right], \nonumber \\
q_0 &= {b_1}^2 \left[{b_0}-({k_1}+{k_2}) \left(b_2^2 ({k_1}+{k_2}-1)-{b_1}\right)\right]. \nonumber
\end{align}
We now request that $f_0(z)$ is proportional to  $z \left(b_1 - b_2^2 z \right)^2$ with some constant multiplier $q$, i.e.
\begin{equation*}
f_0(z) = q z \left(b_1 - b_2^2 z \right)^2.
\end{equation*}
Solving this equation term by term in powers of $z$, we obtain
\begin{align*}
k_1 &= -\frac{q}{b_2^2}, \quad {k_2}=\dfrac{q ({b_1} + q) - b_2^2 ({b_0}+q)}{b_2^2 {b_1}}, \quad q = \dfrac{1}{2} \left(b_2^2 - b_1 \pm \sqrt{b_2^4 + 2 b_2^2 (2 {b_0} + {b_1}) + {b_1}^2}  \right).
\end{align*}
Accordingly, substituting these definitions into \eqref{G} one finds
\begin{align*}
0 &= z G''(z) + (b + z) G'(z) - a G(z), \\
b &= 2 - \frac{b_1 + 2 q}{b_2^2}, \quad a = \frac{q}{b_2^2}. \nonumber
\end{align*}
This is a sort of Kummer equation which has two independent solutions, \cite{PolyaninSaitsevODE2003}
\begin{equation}
G(z) = e^{-z} U(a+b,b,z), \quad G(z) = e^{-z} M(a+b,b,z).
\end{equation}
Accordingly, as $q$ can take two values corresponding to the plus and minus sign, we have four fundamental solutions of the original equation \eqref{G}.

Similar to the previous section, we cannot use these solutions at
the first $0 \le x \le  x_1$ and the last $x_{n_j} < x < \infty$ strike intervals for every smile $T = T_j$ as it violets Lee's formula. Therefore, at these two intervals we use the model discussed in Section~\ref{locVarLog} where the local variance is linear in the log-strike. Accordingly, the local volatility is a square root of the local variance.

\section{Computation of the source term} \label{SolutionInt}

Computation of the source term $p I_{12}$ in \eqref{solInhom} could be achieved in several ways. The most straightforward one is to use numerical integration since the Put price $P(x,T_{i-1})$ as a function of $x$ is already known when we solve \eqref{solInhom} for $T=T_i$. We underline that this is not the case in \cite{ItkinLipton2017}, because there the function $P(x,T_{i-1})$ is obtained by using an inverse Laplace transform, and as such is known only for a discrete set of strikes at the previous time level. Therefore, some kind of interpolation is necessary to find the local variance at all strikes when doing integration. Moreover, this interpolation must preserve no-arbitrage, see \cite{ItkinLipton2017}.

On the other hand, using no-arbitrage interpolation provides another advantage,
as it makes it possible to compute the source term integrals in closed form if the interpolating function is wisely chosen.
Here we want to exploit the same idea, thus significantly improving computational performance of our model as compared with the numerical integration.

Below as an example consider the case of the local variance piecewise linear in the strike space. Then based on solutions found in Section~\ref{locvar} in \eqref{homog2} we have
\begin{align} \label{J1}
J_1(x) &= -y_2(x) \int \dfrac{ y_1(x) c(x)} {(b_2 + a_2 x)x^2 W(x)}d x =
-y_2(x) \frac{a_2^2}{b_2^3} \int \dfrac{ y_1(z) c(z)} {(1-z) z^2 W(z)}d z, \\
y_1(z) &= z^{m} \pFq{2}{1}{m-1, m, c; z}, \quad
c(z) = V(S,T_{j-1}, z), \quad z = - a_2 x/b_2,
\nonumber
\end{align}
\noindent where $W(z)$ is defined in \eqref{wron1}.

Following the idea of \cite{ItkinLipton2017}, in \cite{ELVG}) we introduced a non-linear interpolation
\begin{align} \label{linNew}
P(x) &= \gamma_0 + \gamma_2 x^2, \quad x_1 \le x \le x_3, \\
\gamma_0 &= \dfrac{P(x_3) x_1^2   - P(x_1) x^2_3}{x_1^2 - x_3^2}, \qquad
\gamma_2 = \dfrac{P(x_1) - P(x_3)}{x_1^2 - x_3^2}. \nonumber
\end{align}
Then Proposition~6.1 in \cite{ELVG}) proves that this interpolation scheme is arbitrage-free.

It is worth emphasizing that the proposed interpolation doesn't affect the solution values (quotes) at given market strikes since the piecewise interpolator is constructed to exactly match those values. So the interpolation only affects the Put values that are not known, i.e., those with strikes that lie in between the given market strikes. Therefore, if these strikes are not used, i.e. in trading or hedging, the influence of the interpolation is unobservable at all.  If, however, they are used for some purpose, the difference with the exact solution is small (within the error of interpolation), while the approximate solution for these strikes yet preserves no-arbitrage.

Recall, that  we introduced $V(x)$ using \eqref{P2V}. Accordingly, the term $c(z)$ in \eqref{J1} takes the form (see \ref{App4} and \eqref{intZ})
\begin{equation} \label{cx}
c(z) = V(S,T_{j-1}, z) = \bar{\gamma}_0 + \gamma_1 z + \bar{\gamma}_2 z^2.
\end{equation}
It turns out that now the integral in \eqref{J1} can be computed in closed form. Indeed
\begin{align} \label{stI}
\int & \dfrac{ y_1(z) c(z)} {(1-z) z^2 W(z)}d z = I_0 + I_1 + I_2, \\
I_0 &= \gamma_0 \int \dfrac{ y_1(z)}{(1-z) z^2 W(z)} d z = \bar{\gamma}_0 A(z)
\dfrac{1}{\Gamma(c) (c - m - 1)} \pFq{2}{1}{c-m-1,c-m+1,c,z},
\nonumber \\
I_1 &= \gamma_1 \int \dfrac{ z y_1(z)}{(1-z) z^2 W(z)} d z = \gamma_1 z A_(z) \dfrac{1}{\Gamma(c)(c-m)} \pFq{2}{1}{c-m,c-m,c,z}
\nonumber \\
I_2 &= \bar{\gamma}_2 \int \dfrac{ z^2 y_1(z)} {(1-z) z^2 W(z)}d z = \bar{\gamma}_2 A(z)
z^2 \dfrac{1}{(c - m + 1)\Gamma(c)} \dopFq{3}{2}{c - m,c - m + 1,c - m + 1}{c,2 + c - m}{z}, \nonumber \\
A(z) &= \frac{b_2}{a_2}\Gamma (c-1) z^{c - m -1}, \nonumber
\end{align}
\noindent where $\dopFq{3}{2}{a_1,a_2,a_3}{b_1,b_2}{z}$ is a generalized Hypergeometric function (\cite{Askey2010}).

The second integral in the definition of $J_2$
\begin{align} \label{J2}
J_2(x) &= y_1(x) \int \dfrac{ y_2(x) c(x)} {(b_2 + a_2 x)x^2 W(x)}d x =
y_1(x) \frac{a_2^2}{b_2^3} \int \dfrac{ y_2(z) c(z)} {(1-z) z^2 W(z)}d z, \\
y_2(z) &= z^{m+1-c} \pFq{2}{1}{m - c, m +1 - c, 2 - c; z}, \nonumber
\end{align}
\noindent could be computed in a similar way. The result reads

\begin{align}
\int & \dfrac{ y_2(z) c(z)} {(1-z) z^2 W(z)}d z = {\cal I}_0 + {\cal I}_1 + {\cal I}_2, \\
{\cal I}_0 &= \gamma_0 \int \dfrac{ y_2(z)}{(1-z) z^2 W(z)} d z = \bar{\gamma}_0 A(z)
\frac{1}{m} \pFq{2}{1}{2-m,-m,2-c,z}, \nonumber \\
{\cal I}_1 &= \gamma_1 \int \dfrac{ z y_2(z)} {(1-z) z^2 W(z)}d z =
\gamma_1 A(z) z \frac{1}{(m-1)}\pFq{2}{1}{1-m,1-m,2-c,z}, \nonumber \\
{\cal I}_2 &= \bar{\gamma}_2 \int \dfrac{ z^2 y_2(z)} {(1-z) z^2 W(z)}d z = \bar{\gamma}_2 A(z) z^2 \dfrac{1}{(m-2)} \dopFq{3}{2}{1-m,2-m,2-m}{2-c,3-m}{z}, \nonumber \\
A(z) &= \frac{b_2}{a_2}\frac{\Gamma (1-c)}{\Gamma(2-c)} z^{-m}, \nonumber
\end{align}

Two special cases are the first $ 0 \le x \le x_1$ and the last $x_{n_j} < x < \infty$ intervals where the solution is given by \eqref{Kummer0} and \eqref{Kummer1}.

\subsection{Last interval $x_{n_j} \le x < \infty$.}

Since the right edge of this interval lies at infinity, the interpolation scheme in \eqref{linNew} should be slightly modified. This could be done twofold. The first option is to move the boundary from infinity to any very large but finite positive strike. Then the scheme in \eqref{linNew} could be used with no problem. But in our case it turns out that we are not able to compute these integrals in closed form. Therefore, we use another option which consists in replacing the quadratic form in \eqref{linNew} with another nonlinear interpolation
\begin{equation} \label{int2}
c(\chi) = V(\chi, T_{j-1}, S) = \gamma_\infty z^{-\nu}, \quad z = \chi + \frac{b_2}{a_2},
\end{equation}
\noindent where $\gamma_\infty > 0, \ \nu > 0$ are some constants to be determined. Obviously, at $\chi \to \infty$ this interpolation preserves the correct boundary value of $V$ as in \eqref{bc}, i.e. $V(\chi)$ vanishes in this limit. Derivation of the appropriate values of $\gamma_\infty,\ \nu$ and a proof that the proposed interpolation preserves no-arbitrage are given in \ref{App2}.

Recall that at this interval we assume the local variance to be linear in the log-strike $\chi$. Therefore,  the numerically stable pair of solutions of \eqref{solInhom} is given in \eqref{Kummer1}.  Then the integral in \eqref{solInhom} can be computed in closed form. In doing so we use the following notation from \cite{kummerInt1970}
\begin{align*}
\int e^{- \alpha z} z^\nu U(a,b,z) d z &= U_\nu(\alpha; a,b,z), \\
\int e^{- \alpha z} z^\nu M(a,b,z) d z &= M_\nu(\alpha;
a,b,z).
\end{align*}
Then
\begin{align} \label{stInf}
I_{12}(\chi) &= y_2(\chi) \int \dfrac{ y_1(\chi) c(\chi)}{(b_2 + a_2 \chi)W}d \chi - y_1(\chi) \int  \dfrac{ y_2(\chi) c(\chi)}{(b_2 + a_2 \chi)W}d \chi, \\
\int \dfrac{ y_1(\chi) c(\chi)}{(b_2 + a_2 \chi)W}d \chi &=
\xi_\infty \int e^{-z} z^{-\nu} U(\alpha_1, \beta_1, z) d z = \xi_\infty U_{-\nu}(-1; \alpha_1, \beta_1, z), \nonumber \\
\int \dfrac{ y_2(\chi) c(\chi)}{(b_2 + a_2 \chi)W}d \chi &= \xi_\infty
\int z^{-\nu} U(\beta_1-\alpha_1, \beta_1, -z) d z =
(-1)^{-\nu} \xi_\infty U_{-\nu}(0; \beta_1 - \alpha_1, \beta_1, -z), \nonumber \\
\xi_\infty &= (-1)^{\beta_1-\alpha_1} \gamma_\infty a_2^{2-\beta_1}. \nonumber
\end{align}

As per \cite{kummerInt1970},
\begin{align}
M_{\nu}(-1; a, b, z) &= e^{i \pi (\nu+1)} M_{\nu}(0; b- a, b, -z), \\
M_{\nu}(0; a, b, z) &= \frac{z^{\nu+1}}{\nu+1} \dopFq{2}{2}{\nu_1+1, a}{\nu+2, b}{z}, \quad b \ne 0,-1,-2,\ldots, \quad \nu \ne -1,-2,\ldots, \nonumber \\
M_{-1}(0; a, b, z) &= \frac{a}{b}z\ \dopFq{3}{3}{a+1, 1, 1}{b+1, 2, 3}{z} + \log(z), \nonumber \\
U_\nu(\alpha;a,b,z)  &= \frac{\pi}{\sin(\pi b)}
\left[ \frac{M_\nu(\alpha;a,b,z)}{\Gamma(1+a-b) \Gamma(b)} -
\frac{M_{\nu+1-b}(\alpha;1+a-b,2-b,z)}{\Gamma(a) \Gamma(2-b)} \right]. \nonumber
\end{align}
Therefore, all necessary integrals could be expressed in terms of generalized Hypergeometric functions. Alternatively, these integrals could be represented as
\begin{align} \label{MeijerG}
U_{-\nu}(-1; \alpha_1, \beta_1, z) &=
\MeijerG[\bigg]{2}{1}{2}{3}{1,\ 2+\alpha_1-\beta_1-\nu}
{1-\nu,\ 2-\beta_1-\nu,\ 0}{z}, \\
U_{-\nu}(0; \alpha_1, \beta_1, -z) &=
\frac{z^{1-\nu}}{\Gamma(1-\alpha_1)\Gamma(\beta_1-\alpha_1)}
\MeijerG[\bigg]{2}{2}{2}{3}{\nu,\ 1+\alpha_1-\beta_1}
{0,\ 1-\beta_1,\ \nu-1}{-z}, \nonumber
\end{align}
\noindent where $\MeijerG[\bigg]{m}{n}{p}{q}{a_1,...,a_p}{b_1,...,b_q}{z}$ is the Meijer G-function, see \cite{Olver1997}.

It is not difficult to verify that at $K \to \infty$, and so  $z \to \infty$,  the integral $I_{12}(\chi)$ vanishes.

\subsection{First interval $0 \le x \le x_1$.}

Recall that at this interval we assume the local variance to be linear in the log-strike $\chi$. Since at $K \to 0$ we have $\chi \to -\infty$, the numerically stable pair of solutions of \eqref{solInhom} is still given by \eqref{Kummer1}.

However, at this interval we need another interpolation scheme because the previously described schemes don't give rise to tractable integrals. However, this could be achieved by using, e.g., the following nonlinear interpolation
\begin{equation} \label{int3}
c(\chi) = V(\chi, T_{j-1}, S) = \omega_0 e^z/z, \quad z = \chi + \frac{b_2}{a_2},
\end{equation}
\noindent where $\omega_0 < 0$ is a constant to be determined. Obviously, at $K \to 0$, and so $z \to -\infty$,  this interpolation preserves the correct boundary value of $V$ as in \eqref{bc}, i.e. $V(\chi)$ vanishes in this limit. Derivation of the appropriate value of $\omega_0$ and a proof that the proposed interpolation preserves no-arbitrage are given in \ref{App3}.

Now the integral in \eqref{solInhom} can be computed in closed form
\begin{align} \label{stInf0}
I_{12}(\chi) &= y_2(\chi) \int \dfrac{ y_1(\chi) c(\chi)}{(b_2 + a_2 \chi)W}d \chi - y_1(\chi) \int  \dfrac{ y_2(\chi) c(\chi)}{(b_2 + a_2 \chi)W}d \chi, \\
\int \dfrac{ y_1(\chi) c(\chi)}{(b_2 + a_2 \chi)W}d \chi &=
\xi_0\int z^{-1} U(\alpha_1, \beta_1, z) d z =
\xi_0 U_{-1}(0; \alpha_1, \beta_1, z), \nonumber \\
\int \dfrac{ y_2(\chi) c(\chi)}{(b_2 + a_2 \chi)W}d \chi &= \xi_0 \int e^z z^{-1} U(\beta_1-\alpha_1, \beta_1, -z) d z = - \xi_0 U_{-1}(-1; \beta_1 - \alpha_1, \beta_1, z), \nonumber \\
\xi_0 &= (-1)^{\beta_1-\alpha_1} \omega_0 a_2^{-\beta_1}. \nonumber
\end{align}
Representation of functions $U_{-1}(-1; \beta_1 - \alpha_1, \beta_1, z), \ U_{-1}(0; \alpha_1, \beta_1, z)$ via the Meijer G-function is given in \eqref{MeijerG}. Again, it can be easily verified that at $K \to 0$, and so  $z \to -\infty$,  the integral $I_{12}(\chi)$ vanishes.

\subsection{Special case $z \approx 1$ or $|v/b_2| \ll 1$.}

This case occurs when at the interval $[K_i,K_{i+1}]$  for some $i \in [1,n_j]$  coefficients $a_2, b_2$ are such
that either $|1 - z_i| \ll 1$ or $|1 - z_{i+1}| \ll 1$. Suppose, e.g. that $z_{i+1} = 1 + \epsilon$ with $0 < \epsilon \ll 1.$ As shown in the next section, then we can introduce a ghost point $K_*$ such that $z_* = 1 - \epsilon$. So at the interval $[K_*,K_{i+1}]$ we will use the numerically stable solution in \eqref{homog3}, while at the interval $[K_i,K_*]$ - the regular solution in \eqref{homog2}. Same construction could be provided if $z_{i} = 1 - \epsilon$.

At the interval $z \in [1-\epsilon, 1+\epsilon]$ where the values of $z$ are close to singularity of the Hypergeometric function at $z=1$ there are two ways to construct the solution. First, one can build an asymptotic solution using $v/b_2$ as a small parameter, because  at $z \to 1$ we have $v/b_2 = (b_2 + a_2 x)/b_2  = 1-z \to 0$. As shown in \cite{ELVG}, this can be done, e.g., using the method of boundary functions, \cite{VasBut1995}.

Alternatively, it follows from \eqref{homog3} that $y_1(z) \to 1, \ y_2(z) \to 0$ at $z \to 1$. Therefore, these solutions have a regular behavior in the vicinity of $z = 1$. So all we need to do is to propose a suitable no-arbitrage interpolation to make computation of the source term in \eqref{I12v} tractable. This interpolation is constructed in \ref{App4}.

Thus, based on \eqref{J1} and \eqref{homog3} we need to compute 2 integrals
\begin{align} \label{Jz1}
{\cal J}_1(x) &= \int \dfrac{ y_1(z) c(z)} {(1-z) z^2 W(z)}d z, \qquad
{\cal J}_2(x) = \int \dfrac{ y_2(z) c(z)} {(1-z) z^2 W(z)}d z, \\
y_1(z) &= z^{m} \pFq{2}{1}{m-1, m,2m- c; 1-z}, \quad c(z) = V(z,T_{j-1}, S), \nonumber \\
y_2(z) &= z^{m} (1-z)^{c-2m+1} \pFq{2}{1}{c-m+1,c- m,c-2m+2; 1-z}, \nonumber \\
W(y_1(z), y_2(z)) &= \omega_1 (1-z)^{c-2m}z^{2m - c}, \quad \omega_1 = -\frac{a_2(2m-1-c)}{b_2}. \nonumber
\end{align}
The integral ${\cal J}_2(x)$ can be found in closed form, and the result reads
\begin{align} \label{Jz2exp}
{\cal J}_2(x) &= \bar{\gamma}_0 {\cal J}_{2,0}(x) + \gamma_1 {\cal J}_{2,1}(x)  + \bar{\gamma}_2 {\cal J}_{2,1}(x),  \\
{\cal J}_{2,0}(x) &=  \frac{\pi}{\omega_1}  \csc (\pi  c) z^{-m} \Gamma (c-2 m+2) \Big[ \frac{z^{c-1} \, _2{F}_1(c-m-1,c-m+1;c;z)}{(c-m-1) \Gamma(c) \Gamma (1-m) \Gamma (2-m)} \nonumber \\
&+\frac{\, _2{F}_1(2-m,-m;2-c;z)}{m \Gamma(2-c) \Gamma (c-m) \Gamma (c-m+1)}\Big], \nonumber \\
{\cal J}_{2,1}(x) &=
\frac{\pi} {(m-1) \omega_1 } \csc (\pi  c) z^{-m} \Gamma (c-2 m+2) \Big[(\frac{z (c-m) \, _2{F}_1(1-m,1-m;2-c;z)}{\Gamma(2-c)\Gamma (c-m+1)^2} \nonumber \\
&-\frac{z^c \, _2{F}_1(c-m,c-m;c;z)}{(c-m) \Gamma(c) \Gamma (1-m)^2}\Big], \nonumber \\
{\cal J}_{2,2}(x) &=
\frac{\Gamma (c-2 m+2)}{\omega_1  \Gamma (1-m) \Gamma (2-m) \Gamma (c-m) \Gamma (c-m+1)}
\MeijerG[\bigg]{2}{3}{3}{3}{1,1,2}{2-m,\ c-m+1,\ 0}{z}. \nonumber
\end{align}

The integral ${\cal J}_1(x)$ with the use of no-arbitrage interpolation defined in \eqref{intZ} reads
\begin{equation*}
{\cal J}_1(x) = \omega_1^{-1} \int (1-z)^{-c+2 m-1} z^{c-m-2} \, _2F_1(m-1,m;2 m-c;1-z)
(\bar{\gamma}_0 + \gamma_1 z + \bar{\gamma}_2 z^2) dz.
\end{equation*}
This integral can be computed as follows. We remind that $z \in [1-\epsilon, 1+\epsilon], \ |\epsilon| \ll 1$. Therefore, the term $z^k, \ k \in \mathbb{R}$ can be expanded into series around $z=1$ to obtain
\begin{equation*}
z^k = \sum_{i=0}^\infty  (-1)^i {k\choose i} (1-z)^i
\end{equation*}

Then ${\cal J}_1(x)$ takes the form
\begin{align} \label{expJ1}
{\cal J}_1(x) &= \omega_1^{-1} \Bigg\{
\bar{\gamma}_0  \sum_{i=0}^\infty  (-1)^i {c-m-2 \choose i} \int  (1-z)^{i -c+2 m-1} \, _2F_1(m-1,m;2 m-c;1-z) dz \\
&+ \gamma_1  \sum_{i=0}^\infty  (-1)^i {c-m-1 \choose i} \int  (1-z)^{i -c+2 m-1} \, _2F_1(m-1,m;2 m-c;1-z) dz \nonumber \\
&+ \bar{\gamma}_2  \sum_{i=0}^\infty  (-1)^i {c-m \choose i} \int  (1-z)^{i -c+2 m-1} \, _2F_1(m-1,m;2 m-c;1-z) dz \Bigg\} \nonumber \\
&= \omega_1^{-1} \sum_{i=0}^\infty \nu_i \int  (1-z)^{i -c+2 m-1} \, _2F_1(m-1,m;2 m-c;1-z) dz, \nonumber \\
&= \omega_1^{-1} \sum_{i=0}^\infty \frac{\nu_i}{c - i - 2m} (1-z)^{-c+i+2m} \dopFq{3}{2}{m-1,m,2m-c+i}{2m-c,2m+i-c+1}{1-z}, \nonumber \\
\nu_i &= (-1)^i \left[ \bar{\gamma}_0  {c-m-2 \choose i} +
\gamma_1   {c-m-1 \choose i}  + \bar{\gamma}_2  {c-m \choose i} \right]. \nonumber
\end{align}
The exponent $-c+i+2m = i + b_1/b_2$  is always positive if $b_2 > 0$ in the vicinity of $z=1$. According to \ref{App4}, this condition on $b_2$ is valid if  $1 - \epsilon \le z < 1$. Therefore, 2-3 terms in the expansion \eqref{expJ1} provide the sufficient accuracy in computation of the integral. However, this is also true when $1+\epsilon > z > 1$ (and so $b_2$ is negative) which implies that the entire exponent is also negative, at least at low $i$. This is because the behavior of the product $(1-z)^{i -c+2 m} \dopFq{3}{2}{m-1,m,2m-c+i}{2m-c,2m+i-c+1}{1-z}$ is regular even in this case.

In a similar manner the source terms for other models of the local variance/volatility considered in previous sections could be computed in closed form. We leave this exercise to the reader.

\section{Smile calibration for a single term.} \label{calib}

Calibration problem for the local volatility model is described in \cite{ELVG} as well as the construction of the solution for the entire smile. Here we follow the same approach, and, therefore, provide just some short comments specific to the GLVG model. Again, as an example consider the case where the local variance is a piecewise linear function of strike.
Calibration for the other cases considered in Section~\ref{cases} can be done  in a similar manner.

A general calibration problem we need to solve is: given market quotes of Call and/or Put options corresponding to various strikes $\{K\}:= K_j, \ j \in [1,N]$ and same maturity $T_i$, find the local variance function $v(x)$ such that these quotes solve equations in \eqref{finDupPutM}, \eqref{finDupCallM}.

Suppose that the Put prices for $T=T_j$ are known for $n_j$ ordered strikes. The location of those strikes on the $x$ line is schematically depicted in Fig.~\ref{Fig3}
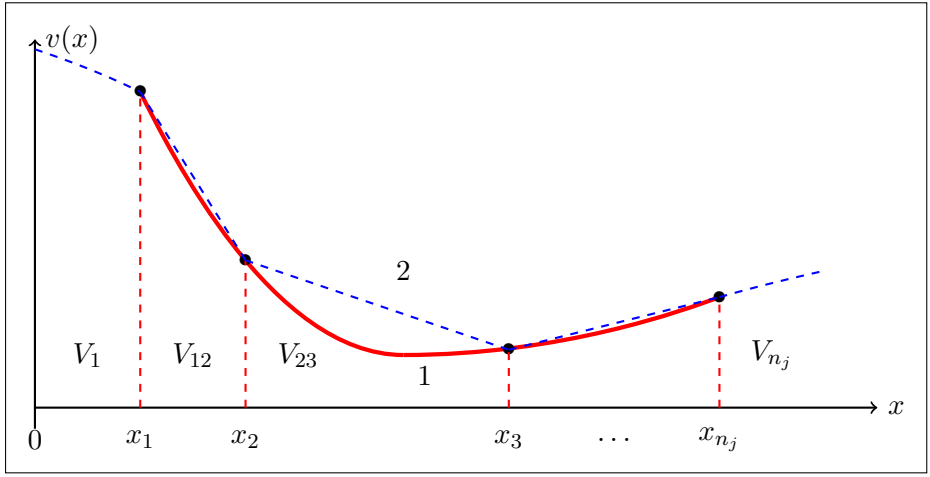
\begin{figure}[!ht]
\begin{center}
\begin{tikzpicture}[line/.style={<->},thick, framed, scale=1.4]

\draw[->] (-3.0,0) -- (5,0) node[right] {$x$};
\draw[->] (-3.0,-0.2) -- (-3.0,3.5) node[right] {$v(x)$};
\draw[red,ultra thick] (0.5,0.5) parabola (3.5,1.05);
\draw[red,ultra thick] (0.5,0.5) parabola (-2,3);
\node at (-3.0,-0.3) {$0$};
\node at (-2,-0.3) {$x_1$};
\node at (-2,3.) {$\bullet$};
\draw[red, dashed] (-2,0) -- (-2,3.);
\node at (-2.5,0.5) {$V_1$};
%
\draw[blue, dashed,domain=-3:-2,samples=20] plot[smooth](\x, {ln(\x/-0.099574)});

\node at (-1,-0.3) {$x_2$};
\node at (-1,1.4) {$\bullet$};
\draw[red, dashed] (-1,0) -- (-1,1.4);
\node at (-1.5,0.5) {$V_{12}$};
\draw[blue, dashed] (-2,3.) -- (-1,1.4);
\node at (1.5,-0.3) {$x_3$};
\node at (1.5,0.55) {$\bullet$};
\draw[red, dashed] (1.5,0) -- (1.5,0.55);
\node at (-0.5,0.5) {$V_{23}$};
\draw[blue, dashed] (-1,1.4) -- (1.5,0.55);
\node at (2.5,-0.3) {$\ldots$};
\node at (3.5,-0.3) {$x_{n_j}$};
\node at (3.5,1.05) {$\bullet$};
\draw[red, dashed] (3.5,0) -- (3.5,1.05);
\draw[blue, dashed] (1.5,0.55) -- (3.5,1.05);
\draw[blue, dashed,domain=3.5:4.5,samples=20] plot[smooth](\x, {ln(\x/1.22478)});
\node at (4,0.5) {$V_{n_j}$};
\node at (0.5,1.3) {$2$};
\node at (0.7,0.3) {$1$};

\end{tikzpicture}
\end{center}
\caption{Schematic construction of the combined solution in $x \in \mathbb{R}^+$: 1 (red solid line) - the real (unknown) local variance curve, 2 (dashed blue line) - a piecewise linear solution. At $x > x_{n_j}$ and $x < x_1$ the blue line is $b_2 + a_2 \log(x)$.}
\label{Fig3}
\end{figure}

Recall that the general form of the solution is given in \eqref{solInhom} which at every interval $x_{i-1} \le x \le x_i$ and  $T=T_j$ can be represented as
\begin{align} \label{combPut}
V(x) &= \CO{i}{j} y_1(x) + \CT{i}{j} y_2(x) + I_{12}(x).
\end{align}
Here for better readability we changed the notation of two integration constants which belong to the $i$-th interval in $x$ and $j$-th maturity to $\CO{i}{j},\CT{i}{j}$.

Similar to \cite{ELVG}, we assume continuity of the options price and its first derivative at every node $i=1,\ldots,n_j$. We also supplement this by two additional conditions: the first one is given by \eqref{contChi}, and the other one is  that at every node the solution $P(S, T_j, K_i)$ must coincide with a given market quote for the pair $(T_j, K_i)$. So together this provides four equations for
four unknown variables $v_{j,i}^0, v_{j,i}^1, \CO{i}{j}, \CT{i}{j}$:
\begin{align} \label{finSystem}
P_i(x)|_{x=x_i} &= P_{i+1}(x)|_{x = x_i}, \\
P_i(x)|_{x=x_i} &= P_{market}(x_i), \nonumber \\
\fp{P_{i+1}(x)}{x}\Big|_{x = x_i} &=  \fp{P_{i}(x)}{x}\Big|_{x = x_i}, \nonumber \\
v^0_{j,i} + v^1_{j,i} x_{i} &= v^0_{j,i+1} + v^1_{j,i+1} x_{i}, \quad i=1,\ldots,n_j. \nonumber
\end{align}
The \eqref{finSystem} is a system of $4 n_j$ nonlinear equations with respect to $4 (n_j+1)$ variables $v_{j,i}^0$, $v_{j,i}^1$, $\CO{i}{j}$, $\CT{i}{j}$. Therefore we need 4 additional conditions to unquietly solve it.

To this end observe that the constants $\CT{1}{j}, \CT{n_j}{j}$ could be determined based on the boundary conditions in  \eqref{bc}. Indeed, at $K \to 0$ function $y_2(\chi)$ in \eqref{Kummer1} vanishes (as $a_2 < 0$ at this interval), but not $y_1(x)$.  Therefore, to obey the vanishing boundary condition in \eqref{bc} we must set $\CO{1}{j} = 0$. As that was already discussed, the source term in \eqref{stInf0} also vanishes in this limit.  Therefore, the solution in \eqref{Kummer1} with the source term in \eqref{stInf0} and $\CT{1}{j} = 0$ obeys the boundary condition at $z \to 0$.

At $K \to \infty$ based on representation of the solution in \eqref{Kummer1} with $a_2 > 0$ at this interval, similarly we must set $\CT{n_j}{j} = 0$, as the solution $y_2(x)$ in \eqref{Kummer1} diverges at $z \to \infty$.

The remaining two additional conditions could be set in many different ways. Here we rely on traders intuition about the asymptotic behavior of the volatility surface at strikes close to zero and infinity. According to our construction, they are determined by $v^1_{j,0}$ and $v^1_{j,n_j}$. Therefore, we assume these coefficients to be somehow known, i.e. consider them as the given parameters of our model.

Overall, by solving the nonlinear system of equations \eqref{finSystem} we find the final solution of our problem. This can be done by using standard methods, and, thus, no any optimization procedure is necessary. However, a good initial guess still would be helpful for a better (and faster) convergence. Construction of such a guess is described in \cite{ELVG}. Also note that this system has a block-diagonal structure where each block is a 2x2 matrix. Therefore, it can be easily solved with the linear complexity $O(n_j)$.

When computing the first derivatives, we take into account that derivatives of Hypergeometric functions belong to the same class of functions, since, \cite{as64}
\begin{align*}
\fp{}{z} \pFq{2}{1}{a,b,c,z} &= \frac{a b}{c}\pFq{2}{1}{a+1,b+1,c+1,z}, \\
\fp{}{z} \dopFq{3}{2}{a,b,c}{d,e}{z} &= \frac{a b e}{c d}\dopFq{3}{2}{a+1,b+1,c+1}{d+1,e+1}{z}.
\end{align*}
Same is true for the Meijer G-function. For instance,
\begin{align}
\fp{}{z} &\MeijerG[\bigg]{2}{2}{2}{3}{\nu,\ 1+\alpha_1-\beta_1}{0,\ 1-\beta_1,\ \nu-1}{-z} =
\frac{\Gamma(1-\alpha_1) \Gamma(\beta_1-\alpha_1)}{z}
U(\beta_1 - \alpha_1, \beta_1, -z) \\
&+ (\nu-1)\MeijerG[\bigg]{2}{2}{2}{3}{\nu,\ 1+\alpha_1-\beta_1}{0,\ 1-\beta_1,\ \nu-1}{-z}. \nonumber
\end{align}

Therefore, computing derivatives of the solution does not cause any new technical problem.

\subsection{Special case $|1-z_i| \ll 1$ at some node $K_i, \ i \in [1,n_j]$.} \label{spCase1}

Without loss of generality suppose that $z_i = 1 - \epsilon$ and $z_{i+1} \gg 1 + \epsilon$ with $0 < \epsilon \ll 1.$ The other case  $z_i = 1 + \epsilon$ and $z_{i-1} \ll 1 - \epsilon$ could be treated in a similar way. Then let us introduce  a ghost point $K_*$ such that $z_* = 1 + \epsilon$. So at the interval $[K_i,K_*]$ we will use the numerically stable solution in \eqref{homog3}, while at the interval $[K_*,K_{i+1}]$ - the regular solution in \eqref{homog2}.

Since $K_*$ is the ghost point, we don't have a market quote available at $K_*$. All we can say is that yet we  assume the local variance/volatility to be a piecewise linear function of $K$ at $[K_*,K_{i+1}]$ and $[K_i,K_*]$. It has to be continuous but with a possible jump in skew at $K_*$.

Since a market quote at $K_*$ is not available, we can replace it with any reasonable value. For instance, an interpolated value between market quotes at $K_i, K_{i+1}$ could be used obtained by using no-arbitrage interpolation\footnote{Despite it looks attractive, we cannot require $v^1_{j,i} = v^1_{j,*}$ since this also gives rise to $v^0_{j,i} = v^0_{j,*}$. However, $v^0_{j,i}$ changes sign at $z = 1$.}. Then we obtain four equations for $\CO{*}{j}, \CT{*}{j, }, v^0_{j,*}, v^1_{j,*}$
\begin{align} \label{finSystem1}
P_i(x)|_{x=x_i} &= P_{*}(x)|_{x = x_i}, \\
P_*(x)|_{x=x_*} &= P_{interp}(x)|_{x = x_*}, \nonumber \\
\fp{P_{*}(x)}{x}\Big|_{x = x_i} &=  \fp{P_{i}(x)}{x}\Big|_{x = x_i}, \nonumber \\
v^0_{j,i} + v^1_{j,i} x_{*} &= v^0_{j,*} + v^1_{j,*} x_{*}, \quad i=1,\ldots,n_j. \nonumber
\end{align}
\noindent that should be added to \eqref{finSystem}. Solving this new combined linear system in the same way as we did it for \eqref{finSystem} we find the values of all unknown $\CO{i}{j}, \CT{i}{j, }, v^0_{j,i}, v^1_{j,i}$ where now $i \in \{ [1,n_j] \cup *\}$.

\section{Discussion} \label{numExp}

First, let us mention that in many practical calculations either coefficients $a_2 = v^1_{j,i}$ at some $i$, or both $b_2 = v^0_{j,i}, \ a_2 = v^1_{j,i}$  (see, for instance,  \eqref{ode2}) are small. Of course, in that case the general solution \eqref{homog2} remains valid. However, when computing the values of Hypergeometric functions numerically, the errors significantly grow in such a case. This is especially pronounced when computing the source term integral $I_{12}$. The main point is that either the Hypergeometric function takes a very small value, and then the constants $\CO{i}{j},\CT{i}{j}$ should be very large to compensate, or vice versa. Resolution of this issue requires a high-precision arithmetics, and, which is more important, taking into account many terms in a series representation of the Hypergeometric functions, which significantly slows down the total performance of the method.

To eliminate these problems we can look at asymptotic solutions of \eqref{ode2} taking into account the existence of small parameters from the very beginning. This approach was successfully elaborated on in \cite{ItkinLipton2017, ELVG}, so we don't describe it here in detail.

In \cite{ELVG} we calibrated the ELVG model, e.g. to the data set taken from \cite{Balaraman2016}. In that paper an implied volatility surface of S\&P500 is presented, and the local volatility surface is constructed using the Dupire formula. We took data for the first 12 maturities and all strikes as they are given in \cite{Balaraman2016}. Our results demonstrated high accuracy and speed of calibration.

When doing so, a technical note should be made. We mentioned already that in our model for every term the slopes of the smile at strikes close to zero, $v^1_{j,0}$ and infinity, $v^1_{j,n_j}$ are free parameters of the model. So often traders have an intuition about these values. However, in our numerical experiments we setup them just using some plausible test values. In particular, in \cite{ELVG} for the sake of simplicity for all smiles we used $v^1_{j,0} = -0.3$, and $v^1_{j,n_j} = 0.1$. Accordingly, for the instantaneous variance $v_j(x_i) = p_j (v^0_{j,i} + v^1_{j,i} \log(x_i))/2$ the slopes at both zero and plus infinity are time-dependent and can be computed by using this definition.

As a numerical solver for the system of linear equations we used the standard Matlab {\it fsolve} function, and utilized a "trust-region-dogleg" algorithm. Parameter "TypicalX" has to be chosen carefully to speedup calculations.

In this paper we repeated this test, but now using the GLVG instead of the ELVG. The results look same as in Fig.5 of \cite{ELVG}, i.e. the quality of the fit is same, and performance of the method is almost same. But the conclusion of \cite{ELVG} remains intact, namely that performance of this model is much better than that reported in both \cite{ItkinSigmoid2015} and \cite{ItkinLipton2017}.

Therefore, a natural question would be: which flavor of the Local Variance Gamma model - arithmetic or geometric one is preferable. Perhaps, if the ultimate goal is fast calibration of the given smile, both could be used interchangeably, and both are capable to provide a good and fast fit. However, for modeling option prices the difference between the geometric and arithmetic LVG models is of the same kind as between the Bachelier and Black-Scholes models. So, for instance, for modeling stock prices the latter would be preferable, while for modeling interest rates the former could provide negative values, which nowadays is a desirable feature.

\section{Conclusions}

In this paper we propose another flavor of the Local Variance Gamma. Several contributions are made as compared with the existing literature. First, the model is constructed based on a Gamma time-changed {\it geometric} Brownian motion with drift, while in all previous papers an arithmetic Brownian motion was used.

Second, we consider 2 models of the local variance - piecewise linear in strike, piecewise linear in the log-strike, and the model of the local volatility piecewise linear in strike (which is new in this context). We also consider a combined model of the local variance which is piecewise linear in strike in the internal intervals, and linear in the log-strike at the first and last intervals (see below in more detail).

Third, we show that for all these new constructions still it is possible to derive an ordinary differential equation for the option price, which plays a role of Dupire's equation for the standard local volatility model. Moreover, it can be solved in closed form in terms of various flavors of Hypergeometric functions. For doing so we propose several new versions of no-arbitrage interpolation, similar to how this was done in \cite{ELVG} but in a slightly different form, so the eintire approach becomes tractable.

Also we shortly discuss various asymptotic solutions which allow a significant acceleration of the numerical solver and improvement of its accuracy in that cases  (i.e, when parameters of the model obey the conditions to apply the corresponding asymptotic). For the sake of brevity we omit the exact derivations as they can be obtained similar to how this is done in \cite{ELVG}.

Fourth,  new boundary conditions are derived for the Put option in the GLVG. They are discrete and converge to the standard boundary conditions in the continuous case (Dupire). These conditions are constructed using some analog of discrete compounding which is natural for the LVG model.

And finally, we notice that for any piecewise model of the local variance/volatility at edge intervals where strikes are close either to 0 or to infinity one has to switch to the local variance linear in log-strike because of Roger Lee's moment formula. Thus, the whole local variance/volatility model becomes a combination of the original model at the internal intervals and local variance linear in log-strike at the edge intervals.

The other features of the GLVG model are pretty much inherited from the ELVG. For instance, similar to \cite{ELVG}, we show that given multiple smiles the whole local variance/volatility surface can be recovered that does not require solving any optimization problem. Instead, it can be done term-by-term by solving a system of non-linear algebraic equations for each maturity, which is faster.

\appendix
\appendixpage

\section{Numerically satisfactory solutions of \eqref{ode2hom} at $z \to \infty$.}  \label{App1}

According to \cite{Olver1997}, the numerically satisfactory fundamental solutions of \eqref{ode2hom} in the vicinity of {\bf singularity at $\boldsymbol{z = \infty}$}
are
\begin{align} \label{solInf}
y_1(x) &= z^m[z^{-A} \pFq{2}{1}{A, A-C+1, A-B+1,1/z}], \\
y_2(x) &= z^m[z^{-B} \ \pFq{2}{1}{B, B-C+1, B-A+1; 1/z}], \nonumber
\end{align}
\noindent where in our case $A = m-1, B = m, C = c$. This substitution transforms the second solution in \eqref{solInf} to
\begin{equation} \label{solInf2}
y_2(x) = z^m[z^{-m} \ \pFq{2}{1}{m, m-c+1, 2; 1/z}],
\end{equation}
\noindent and behaves well at $z \to \infty$. However, since in our setting $n \equiv A - B + 1 = m-1 - m + 1 = 0$, and due to the property
\begin{align*}
\lim_{c \to -n} \dfrac{F(a,b,c;z)}{\Gamma(c)} &= \dfrac{(a)_{n+1}(b)_{n+1}}{(n+1)!}z^{n+1} F(a+n+1,b+n+1,n+2;z), \\
y_1(x) &= F(m-1,m-c,0;z) = \Gamma(0)\dfrac{(m-1)_{1}(m-c)_{1}}{(1)!} z F(m,m-c+1,2;z), \nonumber
\end{align*}
\noindent it turns out that the first solution differs from the second one just by a constant multiplier, i.e. they are not  independent. Therefore, in this case instead the first solution $y_1(x)$  should be chosen based on a more sophisticated analytic continuation of the Hypergeometric function, \cite{bateman1953higher}.
\begin{align} \label{solInf1}
y_1(x) &= z^m[ (-z)^{1-m} \dfrac{\Gamma(c)}{\Gamma(m) \Gamma(c-m+1)} \Psi(z)], \quad |z| > 1, \ |\mathrm{ph}(-z)| < \pi,  \\
\Psi(z) &= 1 - \dfrac{1}{z}\sum^{\infty}_{k=0} \dfrac{(m-1)_{k+1} (m-c)_{k+1}}{k! (k+1)!} z^{-k} \left[\log(-z) + \phi_k)\right], \nonumber \\
\phi_k &\equiv \psi(k+1) + \psi(k+2) - \psi(m+k) - \psi(c - m - k), \nonumber \\
(m)_k &= \Gamma(m)/\Gamma(k), \quad \psi(x) = \Gamma'(x)/\Gamma(x). \nonumber
\end{align}

\section{No-arbitrage interpolation at $\chi \to \infty$.} \label{App2}

In this Appendix we prove the following Proposition:
\begin{proposition} \label{prop2}
Recall that according to \eqref{int2} the proposed interpolation scheme for $V(\chi, T_{j-1}, S)$ at the interval $x_{n_j} \le x < \infty$ reads
\begin{equation}
c(\chi) = V(\chi, T_{j-1}, S) = \gamma_\infty z^{-\nu}, \quad z = \chi + \frac{b_2}{a_2},
\end{equation}
\noindent where $\gamma_\infty > 0, \ \nu > 0$ are some constants determined below in the proof. Also this scheme preserves no-arbitrage.
\end{proposition}
\begin{proof}
By construction, at $K \to \infty$, $c(\chi)$ converges to the correct boundary condition, i.e. vanishes. Assuming that $K_{n_j}$ is in-the-money, \eqref{int2} can be re-written in the form
\begin{equation} \label{PutInf}
P(K) = A(T_{j-1})K - B(T_{j-1})S + \gamma_\infty [\log(K/S) + b_2/a_2]^{-\nu}.
\end{equation}
As at this interval $v = b_2 + a_2 \log(K/S) > 0$, and it was assumed that $K > S$, we must have $a_2 > 0$. Accordingly, to have a positive Put price we require $\gamma_\infty > 0$. This constant could be determined by using a known Put value at $K_{n_j}$, i.e. $P(K_{n_j}) = P_{n_j}$. This yields
\begin{equation}
\gamma_\infty = [P_{n_j} - A(T_{j-1})K_{n_j} - B(T_{j-1})S]\left[\frac{b_2}{a_2} + \log\left( \frac{K_{n_j}}{S}\right) \right]^\nu > 0.
\end{equation}
Therefore, this definition is also consistent with the requirement of positiveness of $\gamma_\infty$.

As this is described in detail in \cite{ItkinLipton2017}, the no-arbitrage conditions for the Put price read
\[ P > 0, \quad P_K > 0, \quad P_{K,K} > 0. \]
Differentiating \eqref{PutInf} on $K$, and then again, we obtain
\begin{align} \label{difK1}
P'_K &= A(T_{j-1}) -  \frac{\gamma_\infty \nu}{K}  \left[\frac{b_2}{a_2} + \log\left( \frac{K}{S}\right) \right]^{-1-\nu},\\
P''_K &= \frac{\gamma_\infty \nu}{a_2 K^2} \left[\frac{b_2}{a_2} + \log\left( \frac{K}{S}\right) \right]^{-\nu-2}
[b_2 + a_2(1+\nu + \log(K/S))]. \nonumber
\end{align}
Analyzing these expressions we conclude that $P''_{K} > 0$. Observe that at $K \to \infty$ we also have $P'_K > 0$. Also observe that $P'_K$ is a monotone function of $K$. Therefore, let us look at $P'_K(K_{n_j})$. Substitution of $K=K_{n_j}$ into the first line of \eqref{difK1} yields
\begin{equation}
P'_K(K_{n_j}) = A(T_{j-1}) + \frac{a_2 \nu}{K_{n_j}(b_2 + a_2 \log(K/S)}
\left[A(T_{j-1}) K_{n_j} - B(T_{j-1})S - P_{n_j}\right].
\end{equation}
As the Put value exceeds its intrinsic value, $P'_K(K_{n_j})$ is positive if
\begin{equation} \label{nuSol}
0 < \nu < A(T_{j-1})  K_{n_j} \left[\frac{b_2}{a_2} + \log\left( \frac{K_{n_j}}{S}\right) \right]
\left[P_{n_j} - A(T_{j-1}) K_{n_j} + B(T_{j-1})S\right]^{-1} \equiv \Omega.
\end{equation}
At large $K_{n_j}$ the expression in the first square brackets is large, and in the second ones - small. Thus the upper boundary for $\nu$ is high enough.

Finally, we take into account the well-known upper bound of the Put option price which is, \cite{hull:97}
\[ P_{n_j} \le A(T_j) K_{n_j}. \]
Because of that, we can re-write \eqref{nuSol} as
\begin{equation} \label{nuSol1}
0 < \nu < \frac{A(T_{j-1})}{B(T_{j-1})}  \frac{K_{n_j}}{S} \left[\frac{b_2}{a_2} + \log\left( \frac{K_{n_j}}{S}\right) \right] \approx \frac{K_{n_j}}{S} \left[\frac{b_2}{a_2} + \log\left( \frac{K_{n_j}}{S}\right) \right] \le \Omega.
\end{equation}
Therefore, if $\nu$ is chosen according to \eqref{nuSol} or \eqref{nuSol1}, this guarantees that $P'_K(K_{n_j}) > 0$. As $P'_K(K)$ is a monotone function of $K$, this proves that with this choice of $\nu$ the condition $P'_K(K) > 0$ is valid at the whole interval $x_{n_j} \le x < \infty$. Thus, this interpolation preserves no-arbitrage.

\bs
\end{proof}

\section{No-arbitrage interpolation at $\chi \to -\infty$.} \label{App3}

In this Appendix we prove the following Proposition:
\begin{proposition} \label{prop1}
Recall that according to \eqref{int2} the proposed interpolation scheme for $V(\chi, T_{j-1}, S)$ at the interval $-\infty \le x < x_1$ reads
\begin{equation}
V(\chi, T_{j-1}, S) = \omega_0 e^{z}/z, \quad z = \chi + \frac{b_2}{a_2},
\end{equation}
\noindent where $\omega_0 = V(\chi_1, T_{j-1}, S) z_1 e^{-z_1} < 0$ is constant. Also this scheme preserves no-arbitrage.
\end{proposition}
\begin{proof}
Obviously, at $K = K_{1}$ we have $\chi_{1} = \log(K_{1}/S), \ V(\chi, T_{j-1}, S) = V(\chi_{1}, T_{j-1}, S) \equiv V_{1}$, therefore, assuming the strike $K_{1}$ is out of the money
\begin{equation} \label{omega}
\omega_0 = V_{1}z_1 e^{-z_1} < 0.
\end{equation}
As this is described in detail in \cite{ItkinLipton2017}, the no-arbitrage conditions for the Put price read
\[ P > 0, \quad P_K > 0, \quad P_{K,K} > 0. \]
Based on \eqref{int2} and the definition of $V$ in \eqref{P2V}, the Put price at this interval can be represented as
\begin{equation} \label{Put0}
P(K, T_{j-1}, S) = \omega_0 e^{z}/z = \omega_0 e^{b_2/a_2} \frac{K/S}{\log(K/S) + b_2/a_2}.
\end{equation}
As at this interval $v = b_2 + a_2 \log(K/S)$, and it was assumed that $K < S$, we must have $a_2 < 0$. Accordingly, to have a positive Put price we require $\omega_0 < 0$. This is consistent with the value of $\omega_0$ introduced in \eqref{omega}.

Differentiating \eqref{Put0} on $K$, and then again, we obtain
\begin{align} \label{difK}
P'_K &= \dfrac{\omega_0 a_2 }{S} e^{b_2/a_2}
\frac{b_2 - a_2 + a_2 \log(K/S)}{(b_2 + a_2 \log(K/S))^2} > 0, \\
P''_K &= -\omega_0 \frac{a_2^2}{K S} e^{b_2/a_2}
\frac{b_2 - 2 a_2 + a_2 \log(K/S)}{(b_2 + a_2 \log(K/S))^3} > 0. \nonumber
\end{align}
Thus, the proposed scheme can be used for interpolation because it provides correct Put option prices at $K = K_1$ and $K \to 0$, and is monotone in $K$. Moreover, it preserves no-arbitrage.
\bs
\end{proof}

\section{No-arbitrage interpolation at $z \to 1$.} \label{App4}

As by definition in \eqref{homog2} $z = -\frac{a_2}{b_2}x$, this implies that
\[ 1 - z = 1 + \frac{a_2}{b_2}x = \frac{v_{ji}}{b_2}. \]
Obviously, $v_{ji} \ge 0$. Therefore, when $z$ is close to 1 two situations are possible:
\begin{enumerate}
\item $ z < 1$, which implies $b_2 > 0$, and accordingly $a_2 < 0$;
\item $ z > 1$, which implies $b_2 < 0$, and accordingly $a_2 > 0$.
\end{enumerate}
Suppose for interpolation of the Put price we use \eqref{linNew}, i.e.
\begin{align} \label{linNew1}
P(x) &= \gamma_0 + \gamma_2 x^2, \quad x_1 \le x \le x_3, \\
\gamma_0 &= \dfrac{P(x_3) x_1^2   - P(x_1) x^2_3}{x_1^2 - x_3^2} =
P_1 - \frac{P_3 - P_1}{x_3^2 - x_1^2}x_1^2 > 0, \qquad
\gamma_2 = \dfrac{P(x_1) - P(x_3)}{x_1^2 - x_3^2} > 0. \nonumber
\end{align}
The second inequality is obvious since $P(x_3) > P(x_1)$ if $x_3 > x_1$. The first one follows from the fact that the Put price exceeds its intrinsic value, i.e.
\[ P_i = [A(T_j)  K_i  - B(T_j) S]^+ + \varepsilon_i, \qquad \varepsilon_i > 0. \]
Suppose, e.g., that  both strikes $K_1, K_3$ are in-the-money. Then
\begin{align}
\gamma_0 &= P_1 - \frac{P_3 - P_1}{x_3^2 - x_1^2}x_1^2 =
P_1 - \frac{A(T_j)S(x_3 - x_1) + \varepsilon_3 - \varepsilon_1}{x_3^2 - x_1^2}x_1^2 \\
&= \frac{P_1 x_3 + x_1(P_1 - A(T_j)K_1)}{x_3 + x_1} + \frac{\varepsilon_1 - \varepsilon_3}{x_3^2 - x_1^2}x_1^2 > 0, \nonumber
\end{align}
\noindent as based on the properties of the Put price $\varepsilon_1 > \varepsilon_3$.

From \eqref{linNew1} it follows that
\begin{align} \label{intZ}
V &= \gamma_0 + \gamma_2 x^2 - A(T_j)S  x + B(T_j) S =
\bar{\gamma}_0 + \gamma_1 z + \bar{\gamma}_2 z^2, \\
\bar{\gamma_0} &= \gamma_0 + + B(T_j) S, \quad \gamma_1 = \frac{a_2}{b_2}A(T_j)S, \quad
\bar{\gamma}_2 = \gamma_2 \frac{a_2^2}{b_2^2}. \nonumber
\end{align}
It was proven in \cite{ELVG} that interpolation \eqref{linNew1} preserves no-arbitrage, and so that in \eqref{intZ}.  We use it when computing ${\cal J}_2(x)$ in \eqref{Jz1}.

\clearpage

\section*{References}

\newcommand{\noopsort}[1]{} \newcommand{\printfirst}[2]{#1}
  \newcommand{\singleletter}[1]{#1} \newcommand{\switchargs}[2]{#2#1}

\end{document}